\documentclass[10pt]{article}
\usepackage{graphicx}%
\usepackage{multirow}%
\usepackage{amsmath,amssymb,amsfonts}%
\usepackage{amsthm}%
\usepackage{mathrsfs}%
\usepackage[title]{appendix}%
\usepackage{xcolor}%
\usepackage{textcomp}%
\usepackage{manyfoot}%
\usepackage{booktabs}%
\usepackage{algorithm}%
\usepackage{algorithmicx}%
\usepackage{algpseudocode}%
\usepackage{listings}%
\usepackage{stmaryrd}
\usepackage{ulem}
\usepackage{url}
\usepackage{multirow}
\usepackage{booktabs}
\usepackage{dsfont}
\usepackage{tikz}
\usepackage{enumitem}
 \usepackage{amsmath}
 \usepackage{amsfonts}

\definecolor{SALMON}{rgb}{0.9725490, 0.4627451, 0.4274510}
\definecolor{TURQUOISE}{rgb}{0.0000000, 0.7490196, 0.7686275}

\usepackage{threeparttable}
\usepackage{tabularx}
\definecolor{ForestGreen}{rgb}{0.07,0.62,0.34}
\usepackage[authoryear,longnamesfirst]{natbib}

\newtheorem{theorem}{Theorem}[section]%  meant for continuous numbers
\newtheorem{proposition}[theorem]{Proposition}% 

\usepackage{hyperref}
\usepackage{authblk}
\usepackage{placeins}

\DeclareMathOperator*{\argmax}{arg\,max}
\DeclareMathOperator*{\argmin}{arg\,min}

\newcommand{\tbf}{{\bf{t}}}
\newcommand{\pibf}{{\bf{\pi}}}
\newcommand{\Psibf}{{\bf{\Psi}}}
\newcommand{\Phibf}{{\bf{\Phi}}}
\newcommand{\thetabf}{{\bf{\theta}}}
\newcommand{\Ibf}{\mathbf{I}}
\newcommand{\ns}{n^{s}}

\hypersetup{
    colorlinks=true,
    linkcolor=black,
    filecolor=magenta,      
    urlcolor=cyan,
    pdftitle={Wang2025study},
    pdfpagemode=FullScreen,
    citecolor=blue,
    breaklinks=true,
    }

\title{Classification of multiple segmented profiles, application to the study of the neuronal protein Tau}

\author[1]{Vincent Brault}

\author[2]{Emilie Lebarbier\footnote{Corresponding author: \url{lebarbie@parisnanterre.fr}}}

\author[3]{Amélie Rosier}

\author[4]{Virginie Stoppin-Mellet}

\affil[1]{\setcounter{footnote}{0}Univ. Grenoble Alpes, CNRS, Grenoble INP\footnote{Institute of Engineering Univ. Grenoble Alpes}, LJK, 38000 Grenoble, France}

\affil[2]{Modal'X, Université Paris Nanterre, Nanterre, France}

\affil[3]{ESME Sudria, Paris, France}

\affil[4]{Univ. Grenoble Alpes, INSERM, U1216, CNRS, CEA, Grenoble Institut Neurosciences (GIN), 38 000 Grenoble, France}

\begin{document}

\maketitle

\abstract{This work is motivated by an application in neuroscience, in particular by the study of the (dys)functioning of a protein called Tau. The objective is to establish a classification of intensity profiles, according to the presence or absence of the protein and its monomer or dimer proportion. For this, we propose a Gaussian mixture model with a fixed number of clusters whose mean parameters are constrained and shared by the clusters. The inference of this model is done via the classical EM algorithm. The performance of the method is evaluated via simulation studies and an application on real data is done.}




\section{Introduction}

%Annals of Applied stat

{\bf{Biological context.}} The Tau protein, found in the brain, plays a crucial role in regulating the architecture and stability of microtubule networks, which are essential for the transmission of nerve signals. Dysfunctions in the Tau protein are linked to serious neurodegenerative disorders, such as Alzheimer’s disease. For a better understanding of how Tau interacts with microtubules and regulates their dynamics, Isabelle Arnal’s team at the Grenoble Institute of Neuroscience has developed an innovative method to reconstitute dynamic microtubule networks in vitro in the presence of Tau protein. This approach allows real-time observation at the single-molecule level using fluorescence microscopy \citep{elie2015tau,stoppin2020studying}.
In these experiments, the Tau protein is labeled with a fluorophore, enabling its detection and tracking as it interacts with microtubules. However, a comprehensive understanding of the reaction kinetics requires determining the stoichiometry of Tau proteins—i.e., the number of molecules involved in the process. This necessitates specific experiments to correlate the detected fluorescence intensity with the number of fluorophores responsible for the signal. Light intensity profiles are generated over time from these experiments. In theory, these profiles resemble stepwise decreasing functions, as shown in Figure \ref{fig:Sauts_theorie}, where each step represents the extinction of a Tau protein. The number of steps, all of equal magnitude, indicates the number of proteins which are present.  \\

%The Tau protein is a protein present in the brain that helps control the architecture and stability of microtubule networks, which are essential for the propagation of nerve messages. Dysfunctions of the Tau protein are responsible for serious pathologies such as Alzheimer's disease. To investigate how Tau controls microtubule dynamics, Isabelle Arnal's team at the Grenoble Neuroscience Institute has developed a method for reconstituting dynamic microtubule networks in vitro in the presence of the Tau protein, and observing them in real time by microscopy \citep{elie2015tau}. Light intensity profiles over time are obtained from these experiments. In theory, these profiles are decreasing step functions, as illustrated in Figure \ref{fig:Sauts_theorie}, where each jump corresponds to the extinction of a protein, and thus the number of jumps (of equal size) indicates the number of proteins which are present. 

\begin{figure}
\begin{center}
\includegraphics[scale=0.37]{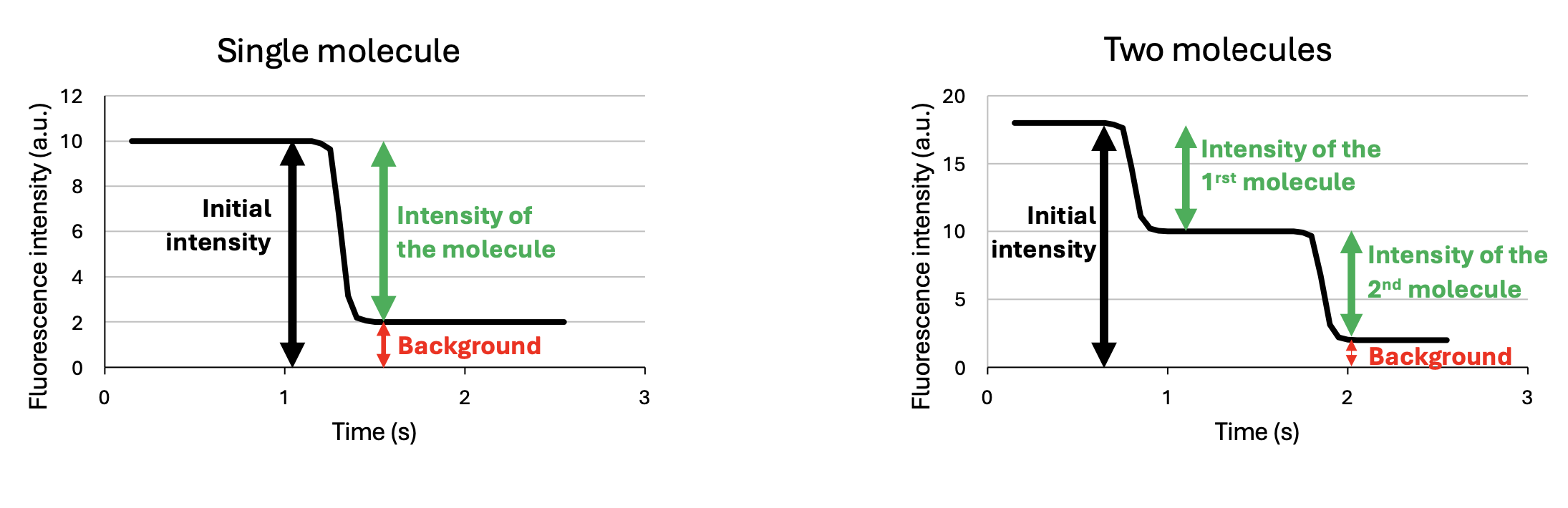}
\caption{Example of two light intensity profiles in theory. In all cases, the jump density (green) is the same, but the initial intensity and background may differ from profile to profile.}
\label{fig:Sauts_theorie}
\end{center}
\end{figure}

\noindent {\bf{Constrained segmentation.}} A natural way to analyze such profile is to use segmentation methods. Indeed they aim at detecting abrupt changes, called change-points, in the distribution of a signal. Segmentation problems arise in many areas such as in biology for the detection of chromosomal aberrations for example \citep{PRL05,LJK05} or in climatology for the detection of instrumental changes \citep{CM04,Mal13}. For our biological problem, a simple segmentation model as a change-point detection in the mean model of a Gaussian process could be used (model that has  been extensively studied in the literature, even recently \citep[see for example][] {HLL2010, PELT, frick2014multiscale, dette2016detecting, arlot2011segmentation}). However, in the biological experiments, proteins that have been switched off may be temporarily reignite before re-extinguishing, an event known as a “blink”. The consequence of such a phenomenon is the apparition of a positive crenel in the profile (with increasing then decreasing mean). The previous model will tend to detect this phenomenon, which is therefore undesirable. This situation is illustrated in Figure \ref{fig:segmentation_free_constrained} on the left on a real profile. One solution is to impose a decay constraint on the mean in the previous model. 

In an algorithmical point of view, it is now well known that segmentation methods have to deal with an inherent algorithmic complexity.
Indeed, the estimation of the change-point parameters requires to search over all the possible segmentations, which is prohibitive in terms of computational time when performed in a naive way. The Dynamic Programming (DP) algorithm \citep{bellman_approximation_1961} is the first algorithm that solved this problem exaclty in a fast way (quadratic with respect to the length of the profile). Unfortunately, DP only applies when the contrast to be
optimized is additive with respect to the segments \citep[see][]{CM04,PRL05}. The decay constraint imposes a link between segments and no longer allows the use of DP. To our knowledge, two papers deals with this problem and propose either a two-stage procedure \citep{gao2020estimation} that consists in doing an isotonic regression then runing the DP on the resulting potential change-points or using a new exact algorithm allowing constraint between successive segments \citep{JSSv101i10}. Note that the two methods recover the exact solution. In Figure \ref{fig:segmentation_free_constrained} on the right, the first method is applied. We can observe that the “blink” is not detected anymore as expected. Note that the decay size cannot be imposed in both strategies, it is estimated in the same time as the estimation of the change-points. \\

\begin{figure}
\begin{center}
\includegraphics[scale=0.15]{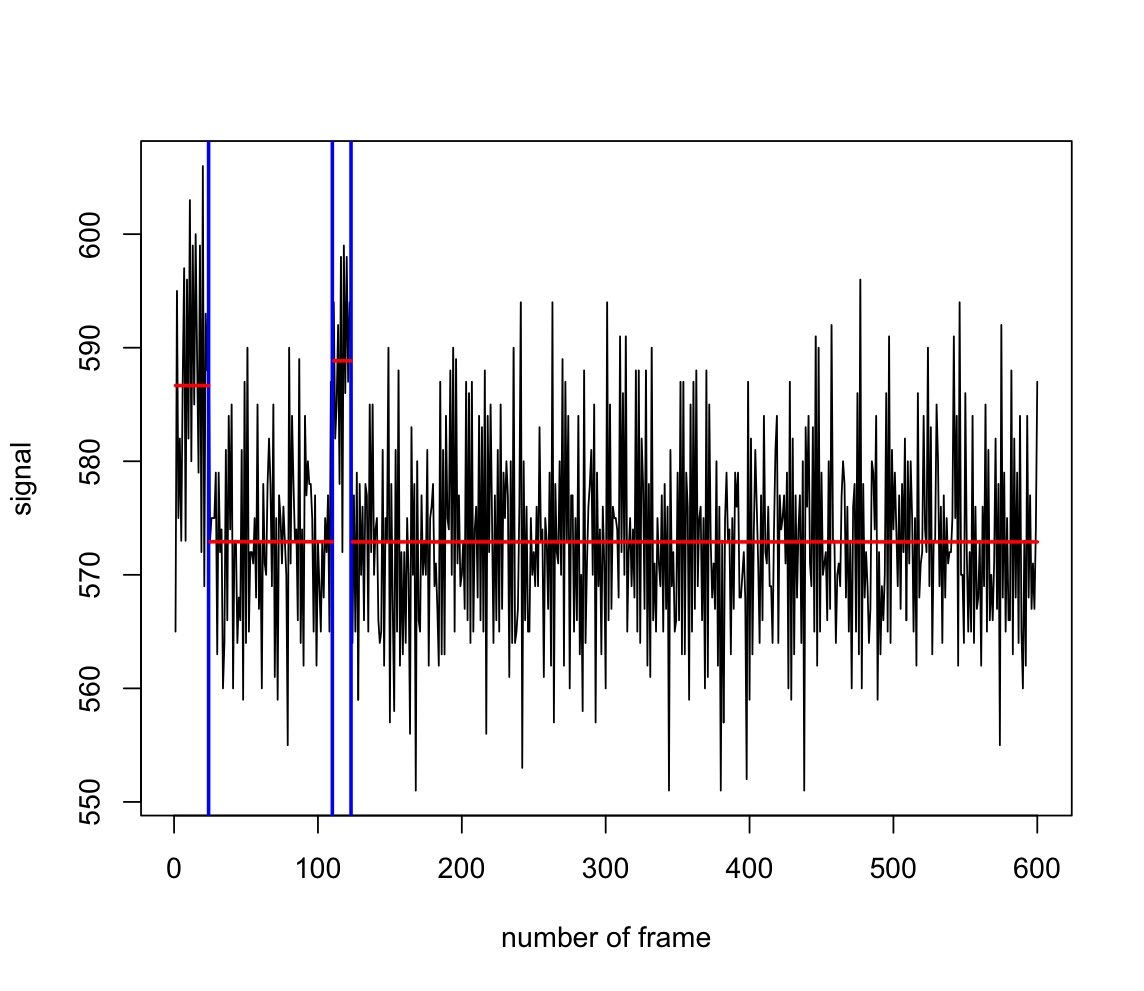}
\includegraphics[scale=0.15]{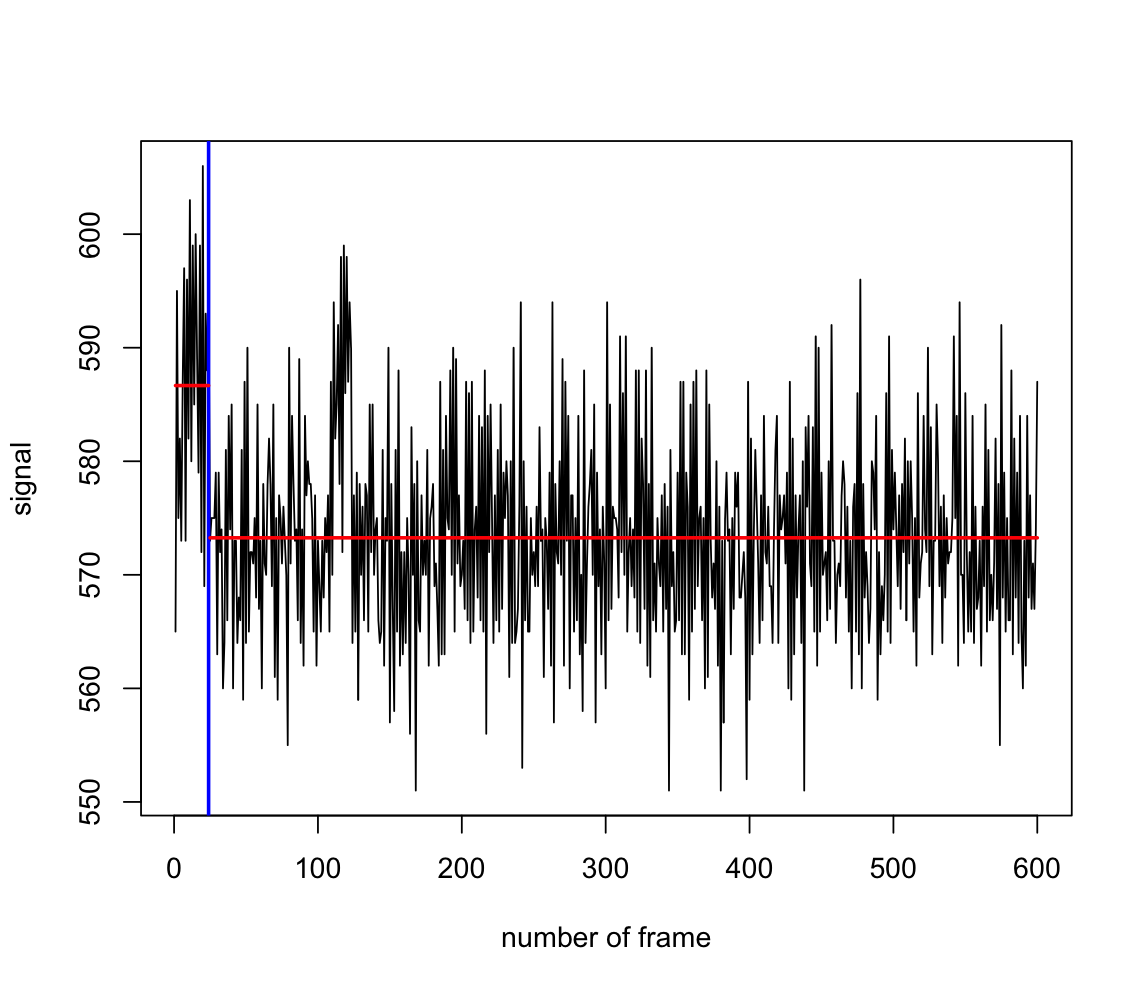}
\caption{Free segmentation (left) and constrained segmentation (right) of a given profile. The red line corresponds to the estimated mean and the vertical blue line(s) to the estimated change-point(s).}
\label{fig:segmentation_free_constrained}
\end{center}
\end{figure}

\noindent {\bf{Constrained classification.}} Our objective in this article is to classify the profiles according to biological constraints: (i) one (or two) decrease(s) in the mean but also (ii) the same decrease in height (same jump) for all profiles (see Figure \ref{fig:Sauts_theorie}). A first and most simple approach to perform this classification could be to use a two-step strategy: segment each profile by considering the decay constraint and cluster the profiles based on the segmentation results (as, for example, the number of jumps and their sizes), as is done classically for times series clustering where the used basis can be Fourier or splines (see \cite{JP14} for a review of these methods). However, constraint (ii) would not be taken into account, which argues for a unified clustering/segmentation model. In this paper, we propose a mixture model of Gaussian distributions whose mean is a decreasing piecewise constant function, differing according to four clusters in order to be in tune with the various biological phenomena: no jump for the cluster $1$, one jump of size $\delta$ for cluster $2$, two jumps of size $\delta$ each for cluster $3$ and one jump of size $2\delta$ for cluster $4$. We impose that the jump heights are shared by all profiles, while the jump times are profile-specific.

To infer the parameters, we propose to use the Expectation/Conditional Maximization algorithm \citep{MR93}. This algorithm is an instance of the traditional EM algorithm \citep{DLR77}, classically considered for the inference of mixture models, which replaces M-step of EM with several computationally CM-steps. Among the CM-steps, one is dedicated to the estimation of the change-point locations, thus once the size of the decay is given (previously estimated). DP-strategies can not still be applied because the segments remain linked through their means. However, here the profiles have a maximum of two change-points, thus a naive search is easily feasible, although a faster strategy would be preferable. \\

\noindent {\bf{Outline.}} In Section \ref{sec:model}, we present the model whose parameters are estimated by maximum likelihood using the ECM algorithm described in Section \ref{sec:inference}. Section \ref{sec:Fisher} discussed the precision of each estimate. Numerical experiments with synthetic data are presented in Section \ref{sec:simulation} and an application on a protein Tau dataset is presented in Section \ref{sec:application}. Section \ref{sec:discussion} presents a discussion.

\section{Model} \label{sec:model}

We consider $S$ profiles with $\ns$ observations each denoting $y^1,\dots,y^S$ with $y^s=(y_{1}^s,\dots,y_{n^s}^s)'$. We suppose that for each profile $s$, $y^s$ is a realization of a Gaussian process $Y^s$ with size $[n^s \times 1]$ whose the number of jumps in the mean depends on the cluster to which the profile belongs, whereas the jump size $\delta$ is the same for all the cluster. More precisely, the mean is either a constant (cluster $1$) or a decreasing piecewise constant function (one jump of size $\delta$ for cluster $2$, two jumps of size $\delta$ each for cluster $3$ and one jump of size $2\delta$ for cluster $4$). We emphasize that the jump location(s) and the baseline mean $\mu^s$ (mean before the first jump) are profile-specific. Formally, for each profile $s$, we introduce a sequence of independent random variables $Z^s=(Z_{1}^s,Z_{2}^s,Z_{3}^s,Z_{4}^s)$ such that for all $k\in[\![1,4]\!]$ 
\begin{displaymath}
Z_k^s=
\begin{cases}
1 & \text{if the profile $s$ belongs to the cluster $k$} \\
0 & \text{otherwise,}
\end{cases}
\end{displaymath}
and
\begin{displaymath}
\pi_k = \mathbb{P}(Z_k^s = 1),
\end{displaymath}
with $\sum_{k=1}^4 \pi_k = 1$. The proportion $\pi_k$ thus represents the prior probability that the profile $Y^s$ belongs to the $k^{th}$ cluster. We assume the between-profile independence (the $S$ profiles $Y^s$ are independent) and the within-profile independence (for each profile $s$, the $(Y^s_t)_{1\leq t\leq \ns}$ are independent), and that
\begin{displaymath}
Y^s \mid Z_{k}^s=1 \sim \mathcal{N}(m_{k}^s,\sigma_k^{2,s} \Ibf_{n^s}) \ \ \ \ \text{for $s\in[\![1\,,S]\!]$,}
\end{displaymath}
where $\Ibf$ is the squared identity matrix and the mean vector of size $[n^s \times 1]$ is for $t \in  \llbracket 1, n^s \rrbracket$ 
\begin{displaymath}
m_{k}^s(t)=
\begin{cases}
\mu_1^s & \text{if $k=1$} \\
\mu_{2}^s+\delta \ \mathbb 1_{t>t_{21}^s} & \text{if $k=2$}\\
\mu_{3}^s+\delta \ \mathbb 1_{t_{31}^s < t \leq t_{32}^s}+ 2 \delta \ \mathbb 1_{t>t_{32}^s} & \text{if $k=3$}\\
\mu_{4}^s+2  \  \delta \ \mathbb 1_{t>t_{41}^s} & \text{if $k=4$}
\end{cases}
\end{displaymath}
or equivalently $m_{k}^s =\mu_k^s +T^s_k \delta,\mathbb 1_{n^s}$ with $\mathbb{1}_n$ the column vector of size $n^s$ with all entries equal to $1$ and where, for profile $s$ in cluster $k$ 
\begin{itemize}
\item $T^s_k$ is the unknown incidence matrix of the change-points with size $[n^s \times 1]$ defined for each cluster as follows
\begin{equation} \label{eq:def_T_k}
T^s_1=\begin{bmatrix} \mathbb{0}_{n^s} \end{bmatrix}, \ \ \ \ T^s_2=\begin{bmatrix}  \mathbb{0}_{n^s_{21}} \\ \mathbb 1_{n^s_{22}} \end{bmatrix}, \ \ \ \ T^s_3=\begin{bmatrix}  \mathbb{0}_{n^s_{31}} \\ 
 \mathbb 1_{n^s_{32}}  \\  2 \mathbb 1_{n^s_{33}} \end{bmatrix} , \ \ \ \ T^s_4=\begin{bmatrix}  \mathbb{0}_{n^s_{41}} \\ 2 \mathbb 1_{n^s_{42}} \end{bmatrix},
\end{equation}
where $\mathbb{0}$ the column vector of size $n^s$ with all entries equal to $0$.
\item $t_{ki}^s$ denotes the $i$th change-point and $n_{ki}^s=t^{s}_{ki}-t^{s}_{k(i-1)}$ is the length of the associated $i$th segment $\llbracket t^{s}_{k(i-1)}+1, t^{s}_{ki} \rrbracket $, with the convention $0=t_{k0}^s<t_{k1}^s<t_{k2}^s<t_{k K_k}^s=n^s$ where $K_k$ is the number of segments in a profile belonging to cluster $k$. 
\end{itemize}
Examples of means $m_{k}^s$ of the four clusters are represented in Figure~\ref{fig:4pops}. \\

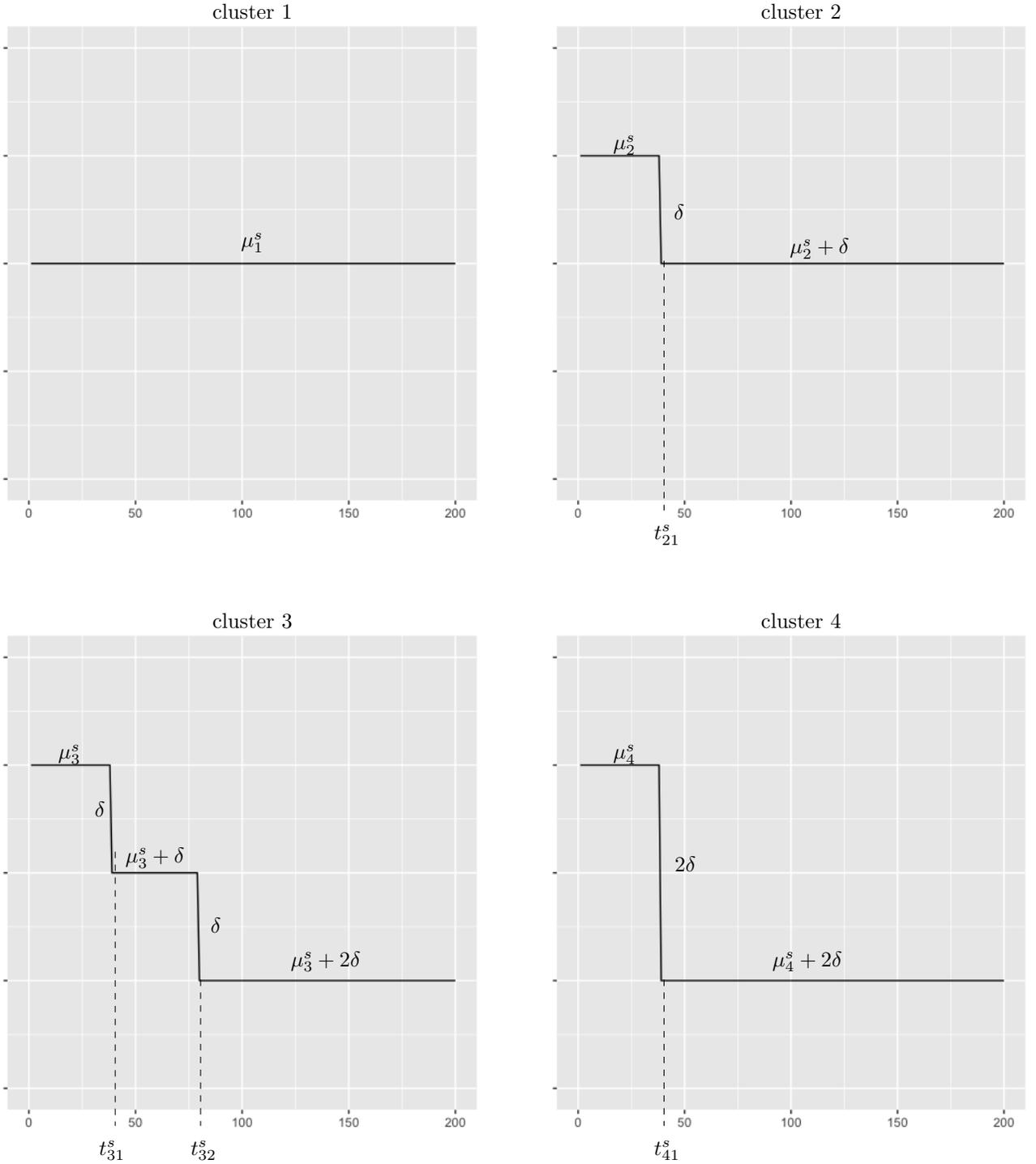
\begin{figure}
\centering
\begin{tikzpicture}
\node at (0,0) {\includegraphics[width=0.5\textwidth]{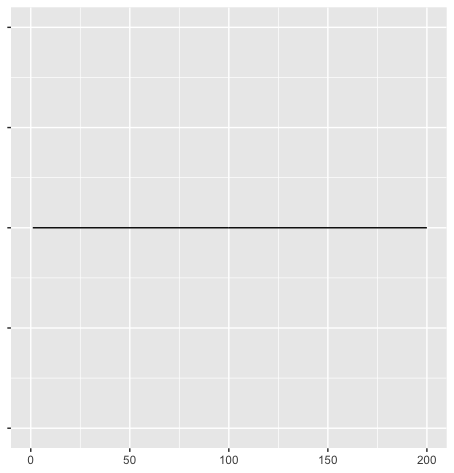}};
 \node at (0.2,4.3) {cluster $1$};
\node at (0.2,.5) {$\mu^s_1$};

\node at (9,0) {\includegraphics[width=0.5\textwidth]{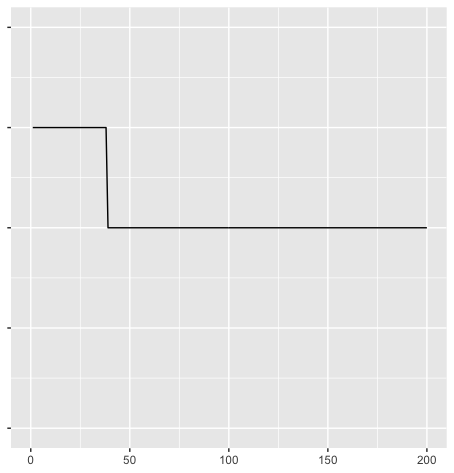}};
\node at (9.2,4.3) {cluster $2$};
\node at (7.2,1) {{$\delta$}};
\node at (6.3,2.1) {{$\mu^s_2$}};
\node at (9.5,0.4) {{$\mu^s_2+\delta$}};
\draw [dashed] (6.95,0.2) -- (6.95,-3.9);
\node at (7,-4.3) {{$t^s_{21}$}};

\node at (0,-10) {\includegraphics[width=0.5\textwidth]{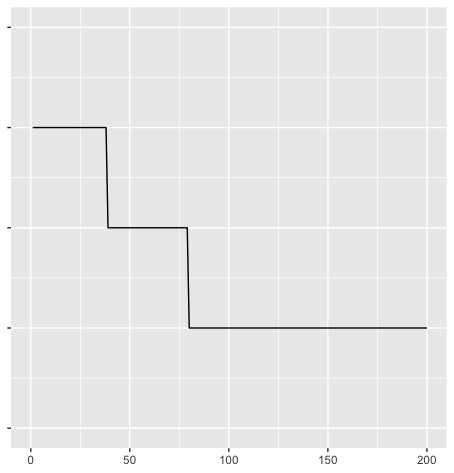}};
\node at (0.2,-5.7) {{cluster $3$}};
\node at (-2.3,-8.8) {{$\delta$}};
\node at (-0.4,-10.7) {{$\delta$}};
\node at (-2.8,-7.9) {{$\mu^s_3$}};
\node at (-1.4,-9.6) {{$\mu^s_3+\delta$}};
\node at (1.4,-11.3) {{$\mu^s_3+2 \delta$}};
\node at (-2.1,-14.4) {{$t^s_{31}$}};
\draw [dashed] (-2.05,-9.5) -- (-2.05,-14);
\node at (-0.6,-14.4) {{$t^s_{32}$}};
\draw [dashed] (-0.65,-11.6) -- (-0.65,-14);
\node at (9,-10) {\includegraphics[width=0.5\textwidth]{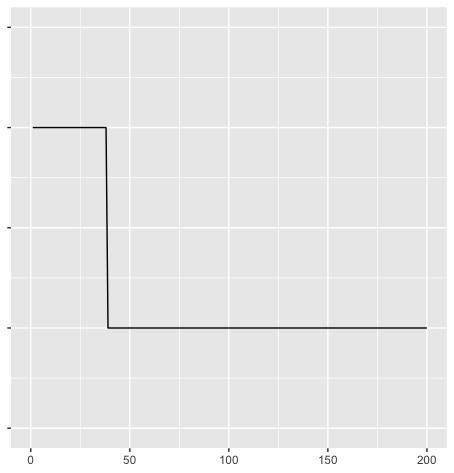}};
 \node at (9.2,-5.7) {{cluster $4$}};
\node at (7.3,-9.7) {{$2 \delta$}};
\node at (6.3,-7.9) {{$\mu^s_4$}};
\node at (9.3,-11.3) {{$\mu^s_4+2 \delta$}};
\node at (7,-14.4) {{$t^s_{41}$}};
\draw [dashed] (6.95,-11.6) -- (6.95,-14);
\end{tikzpicture}
\caption{Theoretical mean of the four clusters. The jump value $\delta<0$ is the same for all profiles and clusters. In this way, the mean value of the profile $s$ at time $t$ depends on cluster belonging $k$, the initial mean $\mu_k^s$ and, except for the cluster $1$, jump location(s) $\left(t^s_{k,j}\right)_{1\leq j\leq K_k}$ (with $K_2=K_4=1$ and $K_3=2$).} \label{fig:4pops}
\end{figure}

The parameters of this model are $\Psibf=(\pibf,\delta)$ with $\pibf=(\pi_k)_{k=1,\ldots,4}$, which involves all the profiles, and $\Phibf=(\Phibf_k^s)_{k,s}$, where $\Phibf_1^s=(\mu_1^s,\sigma_1^{2,s})$ and for $k \geq 2$ $\Phibf_k^s=(\mu_k^s,\sigma_k^{2,s},T^s_k)$, which are profile-specific. We note $\thetabf=(\Psibf,\Phibf)$.\\

\section{Inference}  \label{sec:inference}
We propose to estimate the model parameters by using the classical maximum likelihood method. Our model being a mixture model, we consider the EM algorithm \citep{DLR77} whose use is now well established in this framework. The use of EM lies in the complete-data $\log$-likelihood, that is the $\log$-likelihood of $(Y,Z)$ where $Z=(Z^s)_s$, and which is written
\begin{eqnarray*}
\log \mathbb{P}(Y,Z;\thetabf) 
&=&  \log\left(\prod_{s=1}^S \prod_{k=1}^4 \mathbb ( \pi_k f(Y^s;\Phibf_k^s,\delta))^{Z_k^s}\right) 
= \sum_{s=1}^S \sum_{k=1}^4 Z_k^s \log(\pi_k f(Y^s;\Phibf_k^s,\delta)), 
\end{eqnarray*}
where $f(Y^s;\Phibf_k^s,\delta)$ is the conditional density of $Y^s$ conditionally to cluster $k$, that is
\begin{displaymath}
f(Y_s;\Phibf_k^s,\delta) = \left(\frac{1}{\sigma_{k}^s\sqrt{2\pi}}\right)^{n^s}\exp\left(-\frac{1}{2\sigma_{k}^{2,s}}\| Y^s-m_{k}^s\|_2^2\right)=\left(\frac{1}{\sigma_{k}^s\sqrt{2\pi}}\right)^{n^s}\exp\left(-\frac{1}{2\sigma_{k}^{2,s}}\| Y^s- \mu_{k}^{s} -T^s_k \delta\|_2^2\right),
\end{displaymath}
where $T^s_k$ is given in \eqref{eq:def_T_k}.\\

The EM algorithm aims at maximizing the conditional expectation of the complete-data $\log$-likelihood, given the observed data $Y$, that is 
$$
Q(\thetabf,\thetabf')=\mathbb{E}_{\thetabf'}[\log \mathbb{P}(Y,Z;\thetabf) \mid Y],
$$
where $\mathbb{E}_{\thetabf'}[.]$ denotes the expectation operator using $\thetabf'$ as the parameter value. It is an iterative algorithm which combines two steps: the E-step and the M-step described in the two following subsections respectively. \\

\subsection{E-step.} This step consists in computing the previous conditional expectation using the current fit $\thetabf^{(h)}$ for $\thetabf$:
\begin{eqnarray}\label{Q}
Q(\thetabf,\thetabf^{(h)})&=& \mathbb{E}_{\thetabf^{(h)}} \left[ \sum_{s=1}^S \sum_{k=1}^4 Z_k^s \log( \pi_k f(Y^s;\Phibf_k^s,\delta)) \Bigg| Y\right] = \sum_{s=1}^S \mathbb{E}_{\thetabf^{(h)}} \left[ \sum_{k=1}^4 Z_k^s \log( \pi_k f(Y^s;\Phibf_k^s,\delta)) \Bigg| Y^s\right]  \nonumber \\
&=& \sum_{s=1}^S \sum_{k=1}^4 \mathbb{E}_{\thetabf^{(h)}} \left[Z_k^s| Y^s\right]\log( \pi_k f(Y^s;\Phibf_k^s,\delta)) =  \sum_{s=1}^S \sum_{k=1}^4 {\tau_k^s}^{(h+1)}\log( \pi_k f(Y^s;\Phibf_k^s,\delta)),
\end{eqnarray}
with
\begin{eqnarray*}
{\tau_k^s}^{(h+1)} &=& \mathbb P_{\thetabf^{(h)}}(Z_k^s=1 |Y^s) = \frac{\pi_k^{(h)} f(Y^s;\Phibf_k^{s,(h)},\delta^{(h)})}{\sum_{l=1}^4 \pi_l^{(h)} f(Y^s;\Phibf_l^{s,(h)},\delta^{(h)})},
\end{eqnarray*}
which represents the posterior probability of the profile $Y^s$ to belong to cluster $k$.

\subsection{CM-steps.} This global M-step consists in updating the parameters by maximizing the conditional expectation resulting from the E-step:
\begin{displaymath}
{\thetabf}^{(h+1)} = \argmax_{\thetabf} \ Q(\thetabf,\thetabf^{(h)}). %\sum_{s=1}^S \sum_{k=1}^4 {\tau_k^s}^{(h+1)}\log( \pi_k \ f(Y^s;\Phibf_k^s,\delta)).
\end{displaymath}
We use here the ECM algorithm to break down the maximization of $Q(\thetabf,\thetabf^{(h)})$ with respect to $\thetabf$ into simpler CM-steps which focus on one parameter, the others being fixed. The convergence properties of ECM are provided in \cite{MR93}.\\

\noindent {\bf {Estimation of $\pi$.}} The proportions $\pi_k$ are estimated under the constraint that $\sum_{k=1}^4\pi_k = 1$. We get 
\begin{displaymath}
\pi_k^{(h+1)} = \argmax_{\pibf} \ Q(\thetabf,\pibf,\delta^{(h)},\Phibf^{(h)}) =\frac{\sum_{s=1}^S {\tau_k^s}^{(h+1)}}{S} \ \ \ \text{for $k=1,\ldots,4$.}
\end{displaymath}

\noindent {\bf {Estimation of $\delta$.}} We get
\begin{eqnarray*}
\delta^{(h+1)} &=& \argmax_{\delta} \ 
 Q(\thetabf,\pibf^{(h+1)},\delta,\Phibf^{(h)}), \\
&=& \argmin_{\delta} \sum_{s=1}^S \sum_{k=1}^4 \frac{{\tau_k^s}^{(h+1)}}{\sigma_k^{2, (h)}} \| Y^s- \mu_{k}^{s (h)} -T_k^{s (h)} \delta\|_2^2, \\
&=& \frac{\sum_{s=1}^S \sum_{k=2}^4 \frac{{\tau_k^s}^{(h+1)}}{\sigma_k^{2, (h)}} \ {}^t  T_k^{s (h)} (Y^s- \mu_{k}^{s (h)})}{\sum_{s=1}^S \sum_{k=2}^4 \frac{{\tau_k^s}^{(h+1)}}{\sigma_k^{2, (h)}} \ {}^t T_k^{s (h)} T_k^{s (h)} },
\end{eqnarray*}
where ${}^t H$ denotes the transpose matrix of the matrix $H$. The detailed calculation gives $\delta^{(h+1)}=A/B$ with
\begin{eqnarray*}
A &=&  \sum_{s=1}^S \left\{ \sum_{k=2}^3  \frac{{\tau_{k}^s}^{(h+1)}}{\sigma_{k}^{2,s (h)}} \sum_{i=2}^{K_k} (i-1) \sum_{t=t^{s (h)}_{k (i-1)}+1}^{t^{s (h)}_{ki}}  (Y_t^s-\mu_{k}^{s (h)}) +  2 \frac{{\tau_{4}^s}^{(h+1)}}{\sigma_{4}^{2,s (h)}} \sum_{t=t^{s (h)}_{41}+1}^{n^s} (Y_t^s-\mu_{4}^{s (h)}) \right\},
\end{eqnarray*}
and
\begin{displaymath}
B = \sum_{s=1}^S \left\{ \frac{{\tau_{2}^s}^{(h+1)}}{\sigma_{2}^{2,s (h)}} {n_{22}^{s (h)}} + \frac{{\tau_{3}^s}^{(h+1)}}{\sigma_{3}^{2,s (h)}}({n_{32}^{s (h)}}+4{n_{33}^{s (h)}}) + 4\frac{{\tau_{4}^s}^{(h+1)}}{\sigma_{4}^{2,s (h)}}{n_{42}^{s (h)}}\right\}.
\end{displaymath}

\noindent {\bf {Estimation of $\Phibf$.}} This step is reduced to a constraint segmentation framework in which each profile $Y^s$ need to be segmented in each cluster with an fixed high jump $\delta^{(h+1)}$. For each $s=1,\ldots,S$ and $k=1,\ldots,4$,  
\begin{eqnarray*}
\Phibf_k^{s (h+1)} &= &\argmax_{\Phibf_k^{s}} 
 \  Q(\thetabf,\pibf^{(h+1)},\delta^{(h+1)},\Phibf),
\end{eqnarray*}
i.e. 
\begin{eqnarray*}
(\mu_k^{s (h+1)},\sigma_k^{2,s (h+1)}, T^{s  (h+1)}_{k}) &\in &\argmax_{\mu_k^{s}, \sigma_k^{2,s},T^{s}_{k}} 
 \ Q(\thetabf,\pibf^{(h+1)},\delta^{(h+1)},\mu_k^{s},\sigma_k^{2,s},T^{s}_{k}), \\
&=& \argmax_{\mu_k^{s}, \sigma_k^{2,s},T^{s}_{k}}  \  \left [ -  \frac{n^s}{2} \log{(2 \pi \sigma_k^{2,s})}-\frac{1}{2  \sigma_k^{2,s}} \| Y^s- \mu_{k}^{s} -T^s_k \delta^{(h+1)}\|_2^2 \right ].
\end{eqnarray*}
We differentiate the case of no change (cluster $1$) from the others:
\begin{itemize}
\item for $k=1$, i.e. no change in the mean function, 
\begin{eqnarray*}
(\mu_1^{s (h+1)},\sigma_1^{2,s (h+1)}) &= &\argmax_{\mu_1^{s}, \sigma_1^{2,s}}   \  - \frac{n^s}{2} \log{(2 \pi \sigma_1^{2,s})}-\frac{1}{2  \sigma_1^{2,s}} \| Y^s- \mu_{1}^{s} \|_2^2, 
\end{eqnarray*}
and we obtain 
\begin{equation*}
{\mu_1^s}^{(h+1)} = \overline{Y}^s =\frac{1}{n^s}\sum_{t=1}^{n^s} Y_t^s   \ \ \ \ \text{and} \ \ \ \ \sigma_1^{2,s (h+1)} = \frac{1}{n^s}\sum_{t=1}^{n^s} \left ( Y_t^s-\overline{Y}^s \right )^2.
\end{equation*}
\item for $k \geq 2$, if the change-points $\tbf^{s}_{k}$ are known, the estimation of the mean baseline and the variance for each profile $s$ is straightforward. We get 
\begin{equation*}
%T^{s}_{k}\mapsto{\widehat{\mu}_k^s}(T^{s}_{k}) = \frac{1}{n^s} \sum_{t=1}^{n^s}\left (Y_t^s-\left(T^s_k\right)_t \delta^{(h+1)}\right )  \ \ \ \ \text{and} \ \ \ \ T^{s}_{k}\mapsto{\widehat\sigma_k^{2,s}}(T^{s}_{k}) = \frac{1}{n^s}\| Y^s-{\widehat{\mu}_k^s}(T^{s}_{k})  -T^s_k \delta^{(h+1)}\|_2^2.
\widehat{\mu}_k^s (T^{s}_{k})=\frac{1}{n^s} \sum_{t=1}^{n^s}\left (Y_t^s-\left(T^s_k\right)_t \delta^{(h+1)}\right )  \ \ \ \ \text{and} \ \ \ \ {\widehat\sigma_k^{2,s}}(T^{s}_{k}) = \frac{1}{n^s}\| Y^s-{\widehat{\mu}_k^s}(T^{s}_{k})  -T^s_k \delta^{(h+1)}\|_2^2.
\end{equation*}
The estimation of the change-points is then obtained as follows
\begin{eqnarray*}
T^{s  (h+1)}_k &= & \argmax_{T^{s}_k} \max_{{\mu}_k^s} \max_{\sigma_k^{2,s} } 
 \ Q(\thetabf,\pibf^{(h+1)},\delta^{(h+1)},\mu_k^{s},\sigma_k^{2,s},T^{s}_k) \\
&=&\argmax_{T^{s}_k} 
 \ Q(\thetabf,\pibf^{(h+1)},\delta^{(h+1)},{\widehat{\mu}_k^s}(T^{s}_k),{\widehat\sigma_k^{2,s}}(T^{s}_k),T^{s}_k), \\
&=& 
\argmax_{T^{s}_k} \ \left [ - \frac{n^s}{2} \log{(2 \pi {\widehat\sigma_k^{2,s}}(T^{s}_k))}-\frac{1}{2  {\widehat\sigma_k^{2,s}} (T^{s}_k)}  \| Y^s-{\widehat{\mu}_k^s}(T^{s}_{k})  -T^s_k \delta^{(h+1)}\|_2^2 \right ], \\
%&=& 
%\argmax_{T^{s}_k}  - \frac{n^s}{2} \left [\log{(2 \pi {\widehat\sigma_k^{2,s}}(T^{s}_k))}+\frac{1}{2  {\widehat\sigma_k^{2,s}} (T^{s}_k)}  \times\underbrace{\frac{2}{n^s}\| Y^s-{\widehat{\mu}_k^s}(T^{s}_{k}) \mathbb 1_{n^s} -T^s_k \delta^{(h+1)}\|_2^2}_{=2\widehat\sigma_k^{2,s}(T^{s}_k)} \right ] \\
&=& \argmax_{T^{s}_k} \  -\frac{n^s}{2} \left [ \log{(2 \pi {\widehat\sigma_k^{2,s}}(T^{s}_k))} +1 \right ], \\
&=& \argmin_{T^{s}_k} \  \| Y^s-{\widehat{\mu}_k^s}(T^{s}_{k})  -T^s_k \delta^{(h+1)}\|_2^2, 
\end{eqnarray*}
and the resulting estimators of the mean baseline and the variance are
\begin{equation*}
{\mu}_k^{s (h+1)}={\widehat{\mu}_k^s}(T^{s (h+1)}_k) \ \ \ \ \text{and} \ \ \ \ {\sigma_k^{2,s (h+1)}}
 ={\widehat\sigma_k^{2,s}}(T^{s (h+1)}_k).
\end{equation*}
As an example, let take $k=2$. In this case, the mean is affected by only one change-point $t^s_{21}$. If this change-point is known, we have
\begin{equation*}
{\widehat{\mu}_2^s}(t^{s}_{21}) = \frac{1}{n^s}\sum_{t=1}^{n^s} \left (Y_t^s-\delta^{(h+1)} \mathbb 1_{t>t_{21}^s} \right )  \ \ \ \ \text{and} \ \ \ \ {\widehat\sigma_2^{2,s}}(t^{s}_{21}) = \frac{1}{n^s}\sum_{t=1}^{n^s} \left  (Y_t^s-{\widehat{\mu}_2^s}(t^{s}_{21})-\delta^{(h+1)} \mathbb 1_{t>t_{21}^s}  \right )^2.
\end{equation*}
Then
\begin{eqnarray*}
t_{21}^{s  (h+1)} &= & \argmin_{0 < u < n^s} \sum_{t=1}^{n^s} \left (Y_t^s-{\widehat{\mu}_2^s}(u)-\delta  \mathbb 1_{t>u} \right )^2.
\end{eqnarray*}
\end{itemize}

The final estimators are $\widehat{\pi}_k, \widehat{\delta}, \widehat{\mu}_k^s, \widehat\sigma_2^{2,s}, \widehat{T}_k^s$ for $k=1,\dots,4$ and $s=1,\ldots,S$.

\section{Assessing the precision of estimates} \label{sec:Fisher}

In the missing-data framework where parameters are estimated using the EM algorithm, the observed Fisher information is often used to assess the precision of the estimators. However, because of the presence of latent variables, there is no simple form of it. Instead, the expected complete-data Fisher information matrix, taken with respect to the conditional distribution of the missing data given the observed data, provides a reasonable approximation of the uncertainty. Although it tends to overestimate precision, it remains useful for identifying the relative importance of various factors - such as model specification and sample size - that influence estimator precision \citep{mclachlan2008em, louis1982finding}. \\

We aim at calculating 
\begin{equation*}
\mathcal{I}(\widehat{\theta})= \mathbb{E}_{X | Y} [I^c(\theta)] \big|_{\theta=\widehat{\theta}},
\end{equation*}
where $X=(Y,Z)$ and $I^c(\theta)$ is given by 
\begin{equation*}
I^c(\theta)=- \frac{\partial^2 }{\partial \theta \partial {}^t \theta} \log{\mathbb{P}(Y,Z;\theta)},
\end{equation*}
and so-called the complete-data information matrix. It is straightforward to see that this matrix is almost block diagonal
\begin{equation*}
I^c(\theta)=\begin{bmatrix}
I^c(\pi) & 0 & 0 &0 \\
0 & I^c(\delta) & I^c(\delta,\mu) & I^c(\delta,\sigma^2) \\ 
0  & I^c(\delta,\mu) & I^c(\mu) &0 \\
0 & I^c(\delta,\sigma^2) & 0 & I^c(\sigma^2) \\
\end{bmatrix}.
\end{equation*}
In Appendix \ref{sec:Calcul-Fisher}, we obtain that  
\begin{itemize}
\item for the proportion parameter
\begin{eqnarray*}
\mathbb{E}_{X | Y} [I^c(\pi)]\big|_{\pi=\widehat{\pi}} &=& S  
\begin{bmatrix}
\frac{1}{\widehat{\pi}_1}+\frac{1}{\widehat{\pi}_4} & \frac{1}{\widehat{\pi}_4} & \frac{1}{\widehat{\pi}_4} \\
\frac{1}{\widehat{\pi}_4}& \frac{1}{\widehat{\pi}_2}+\frac{1}{\widehat{\pi}_4} & \frac{1}{\widehat{\pi}_4} \\
\frac{1}{\widehat{\pi}_4}& \frac{1}{\widehat{\pi}_4} & \frac{1}{\widehat{\pi}_3}+\frac{1}{\widehat{\pi}_4}
\end{bmatrix},
\end{eqnarray*}
\item for the jump parameter
\begin{eqnarray*}
\mathbb{E}_{X | Y} [I^c(\delta)]\big|_{\theta=\widehat{\theta}} &=&  \sum_{s=1}^S \sum_{k=2}^4 \frac{ \tau^s_k}{\widehat{\sigma}^{2,s}_k} \ {}^t {T_s^k} T_s^k\\
&=&  \sum_{s=1}^S \left (  \frac{ n_{22}^s\tau^s_2}{\widehat{\sigma}^{2,s}_2} + \frac{ \tau^s_3 (n_{32}^s+4 n_{33}^s)}{\widehat{\sigma}^{2,s}_3} + \frac{ 4 n_{42}^s \tau^s_4 }{\widehat{\sigma}^{2,s}_4}\right ),
\end{eqnarray*}

\item for the mean parameter, $\mathbb{E}_{X | Y} [I^c(\mu)]\big|_{\theta=\widehat{\theta}}$ is a $[4 S \times 4S]$ diagonal matrix with 
\begin{eqnarray*}
\mathbb{E}_{X | Y} [I^c(\mu)]\big|_{\theta=\widehat{\theta}} &=& \text{diag} \left[ \frac{n^s \tau^s_k}{\widehat{\sigma}^{2,s}_k} \right]_{ks},
\end{eqnarray*}

\item and for the variance parameter, $\mathbb{E}_{X | Y} [I^c(\sigma^2)]\big|_{\theta=\widehat{\theta}}$ is a $[4 S \times 4S]$ diagonal matrix with 
\begin{eqnarray*}
\mathbb{E}_{X | Y} [I^c(\sigma^2)]\big|_{\theta=\widehat{\theta}} &=& \text{diag} \left[ \frac{n^s \tau^s_k}{\widehat{\sigma}^{4,s}_k} \right]_{ks}. 
\end{eqnarray*}
\end{itemize}

As expected, the precision of the different estimators is linked to the statistical unit used to estimate them: for the proportion parameter, it only depends on the number of profiles and not on the length of the profiles, contrary for the mean and variance parameters, whose precision depends on and increases with the length of the profiles. For the jump $\delta$, which is common to all clusters, the precision naturally depends on both the number of profiles and their lengths, more precisely on the length of the segments associated to means which include $\delta$. The two terms $I^c(\delta,\mu)$ and $I^c(\delta,\sigma^2)$ are not nulls (not shown here, see Appendix \ref{sec:Calcul-Fisher} for their calculations), indicating that the estimator of the jump is correlated to those of the baseline mean (positively) and the variance, as logical.   \\

Similar results have been found in \citet[chap. 7]{picard2005process} for a different problem, but one that also involves classification and segmentation purposes. His problem is to classify the segments of a segmented profile inferred using a DP-EM algorithm. Similarity in the sense that the concatenation of the whole profile can be viewed as one profile with fixed boundaries, and so the number of profiles in our case plays the role of the number of segments in Picard's model, but the difference is that the distribution parameters are here profile-specific.

\section{Simulation study}  \label{sec:simulation}

In this section, we conduct several simulation studies (1) to study the impact of using several iterations of the M-step or just one, but over several initializations, (2) to study the asymptotic behavior when both the number of profiles and their length increase and (3) to illustrate the importance of the jump estimating. 

\subsection{Simulation design and quality criteria} \label{sec:InitVSIteration:design}

\subsubsection{Simulation design.} We present the basic simulation design and derive some specificities according to the studies carried out.  \\

\noindent {\bf{Basic simulation design.}} We consider $S$ profiles which all have the same length $\ns$. For profile $s$, the probability to belong to each cluster is $1/4$, the positions of the change-points are uniformly distributed on the interval $\llbracket 0.1 \ \ns, \ 0.9 \ \ns \rrbracket$ with a distance of mid$\_$step$=\{0,\  0.3 \ \ns, \ 0.6 \ \ns\}+1$ between the two jumps in cluster $3$ and the baseline mean $\mu_k=2$ for each cluster. We consider different size of jumps $\delta\in\{-5,-2,-1,-0.5\}$. For the estimation, we denote by $Nb_{\text{init}}$ the number of initializations tested at the EM initialization, by $Nb_{\text{M-step}}$ the number of iterations performed within the M-step and by $Nb_{\text{max-EM}}=100$ the maximal number of iterations (one iteration corresponds to both E step and M-step) of the algorithm.\\

\noindent {\bf{Study 1: impact of $Nb_{\text{init}}$ and $Nb_{\text{M-step}}$.}} We consider a fixed number of profiles $S=100$ with length $\ns=100$ and we study the quality of the estimates for different values of $Nb_{\text{init}}=1$ or $10$ random initializations, only the best being kept in the sense of maximum likelihood, and $Nb_{\text{M-step}}=1$ or $10$. In the results, this is codified with the couple type=$(Nb_{\text{init}},Nb_{\text{M-step}})$. Thus, for couples worth $(10,1)$ and $(1,10)$, the maximum number of iterations performed in total by the procedure is the same and is equal to $Nb_{\text{init}}\times Nb_{\text{M-step}}\times Nb_{\text{max-EM}}=1000$; to compensate, $Nb_{\text{max-EM}}=1000$ is taken for the couple $(1,1)$. From this study, a couple is chosen for the two following.  \\

\noindent {\bf{Study 2: asymptotic behavior.}} We examine if the quality of the estimates improves as we increase the number $S$ of profiles and their length $\ns$. We consider for both the values $\{100,250,500\}$. \\

\noindent {\bf{Study 3: importance of jump estimating.}} We study the impact of an underestimation or an overestimation of the jump $\delta$ on the classification. To this aim, for an original jump $\delta$, we fix it in the classification procedure to $\delta'\in\{0.8 \ \delta,\ 0.9 \ \delta,\ \delta,\ 1.1\ \delta,\ 1.2\ \delta,\ 1.5\ \delta,\ 2 \ \delta\}$. We consider $S=100$ and $\ns=250$.

\subsubsection{Quality criteria.} The performances are assessed according to the following criteria:
\begin{itemize}
  \item the relative distance between the estimated and the true jump $d_r(\widehat{\delta},\delta)=\frac{\widehat{\delta}-\delta}{|\delta|}$;
 
  \item the normalized infinite distance between the true change-point locations and the estimated ones in the true cluster. More precisely, for a profile $s$ in cluster $k$, we denote $t^{s(k)}=(t_{kh}^s)_h$ the true change-points and $\widehat{t}^{s(k)}=(\widehat{t}^s_{k \ell})_{\ell}$ the estimated change-points in cluster $k$, the normalized infinite distance is over all the profiles, for each cluster $k\in\{2,3,4\}$, 
\begin{align*}
    d\left(\bold{t}^{\left(k\right)},\widehat{\bold{t}}^{\left(k\right)}\right)=\frac{1}{Z_k^+}\underset{\text{with }Z_k^s=1}{\sum_{s=1}^S}\frac{d_{\infty}(t^{s\left(k\right)},\widehat{t}^{s\left(k\right)})}{\ns}
\end{align*}
where $Z_k^+=\sum_{s=1}^SZ_k^s$ represents the number of profiles $s$ belonging to cluster $k$ and for each $s$
\begin{align*}
d_{\infty}(t^{s(k)},\widehat{t}^{s(k)})=  \left|t_{k1}^s- \widehat{t}_{k1}^s\right|\text{ for }k\in\{2,4\}\quad\text{ and }\quad d_{\infty}(t^{s(3)},\widehat{t}^{s(3)}) =  \max\left(\left|t_{31}^s- \widehat{t}_{31}^s\right|,\left|t_{32}^s- \widehat{t}_{32}^s\right|\right).
\end{align*}
 % conditionally to the true cluster $d(\tau_k,\widehat{\tau}_{\widehat{k}} )$ between the normalized true change-point locations ($\tau_k=(\tau_{kl})_{kl}$ where $\tau_{kl}=t_{kl}/\ns$) and the normalized estimated ones, abusively denoting 
  %$T_k$ and $\widehat{T}_{\widehat{k}}$ respectively, when the profile is classified in the true cluster. More specifically, defining
 % \begin{align*}
 % d_1(T_k,\widehat{T}_{\widehat{k}}) &= \max_h \min_\ell |t_{kh}- \widehat{t}_{\widehat{k} \ell}|, \qquad
 % d_2(T_k,\widehat{T}_{\widehat{k}}) = d_1(\widehat{T}_{\widehat{k}},T_k), \\
 % d(T_k,\widehat{T}_{\widehat{k}})  &= \max(d_1(T_k,\widehat{T}_{\widehat{k}}), d_2(T_k,\widehat{T}_{\widehat{k}})),
%  \end{align*}
%  $d_1$ indicates if each true change-point $t_{k h}$ is close to an estimated one $\widehat{t}_{\widehat{k} \ell}$ ($d_1$ will typically be small when $\widehat{K} \gg K$), whereas $d_2$ indicates if each estimated change-point $\widehat{t}_{\widehat{k} \ell}$ is close to a true one $t_{k h}$. A perfect segmentation results in both null $d_1$ and $d_2$ (and $d = 0$); 
  \item the misclassification rate.
\end{itemize}

%The quality criteria are here normalized by the number of individuals $S$ for clusters and the number of observations $\ns$ for infinite distances. 

\subsection{Results }

\subsubsection{Study 1.} %The accuracy of the segmentation is evaluated first, followed by the classification.

\noindent {\bf{Accuracy of the common jump and the estimated change-points.}} Figures \ref{fig:Delta_norm} and \ref{fig:Hausdorff} represent the distribution of the relative distance $d_r(\widehat{\delta},\delta)$ and the normalized infinite distance $d\left(\bold{t}^{\left(k\right)},\widehat{\bold{t}}^{\left(k\right)}\right)$, respectively. We observe that when the detection problem is easy ($\delta$ large), the jump is well estimated and the change-points are well positioned (the infinite distance is close to $0$), while when the detection problem becomes difficult ($\delta$ small), the jump is underestimated and the change-points are less well positioned, particularly for cluster $3$. This observation is expected.\\

\begin{figure}
    \centering
    \includegraphics[width=\linewidth]{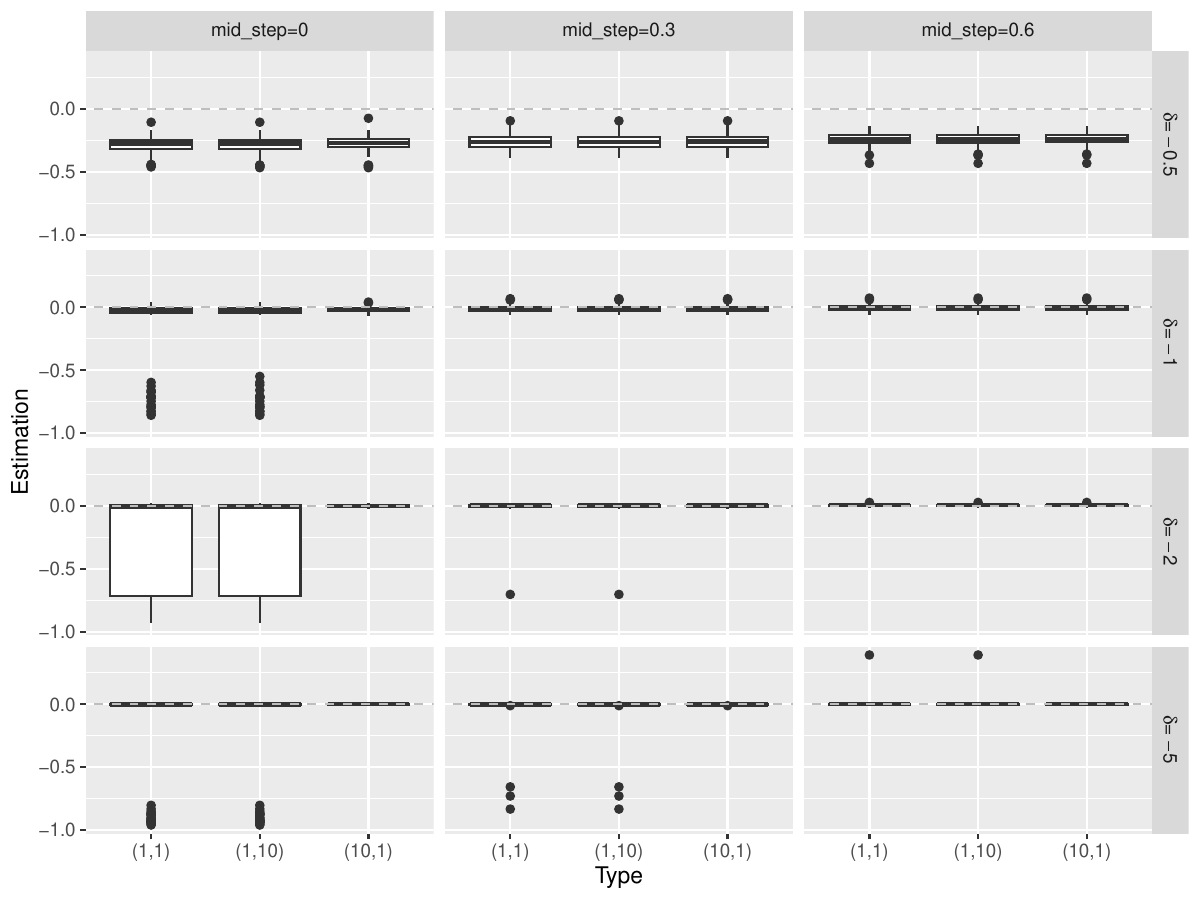}
    \caption{Boxplot of $d_r(\widehat{\delta},\delta)$ for the different values of $\delta$ (lines), the different length between the two change-points in cluster $3$ (columns) and different values for the couple type=$(Nb_{\text{init}},Nb_{\text{M-step}})$ where $Nb_{\text{init}}$ is the number of initializations and $Nb_{\text{M-step}}$ is the number of iterations in the M-step (x-axis).}
    \label{fig:Delta_norm}
\end{figure}

\begin{figure}
    \centering
    \includegraphics[width=\linewidth]{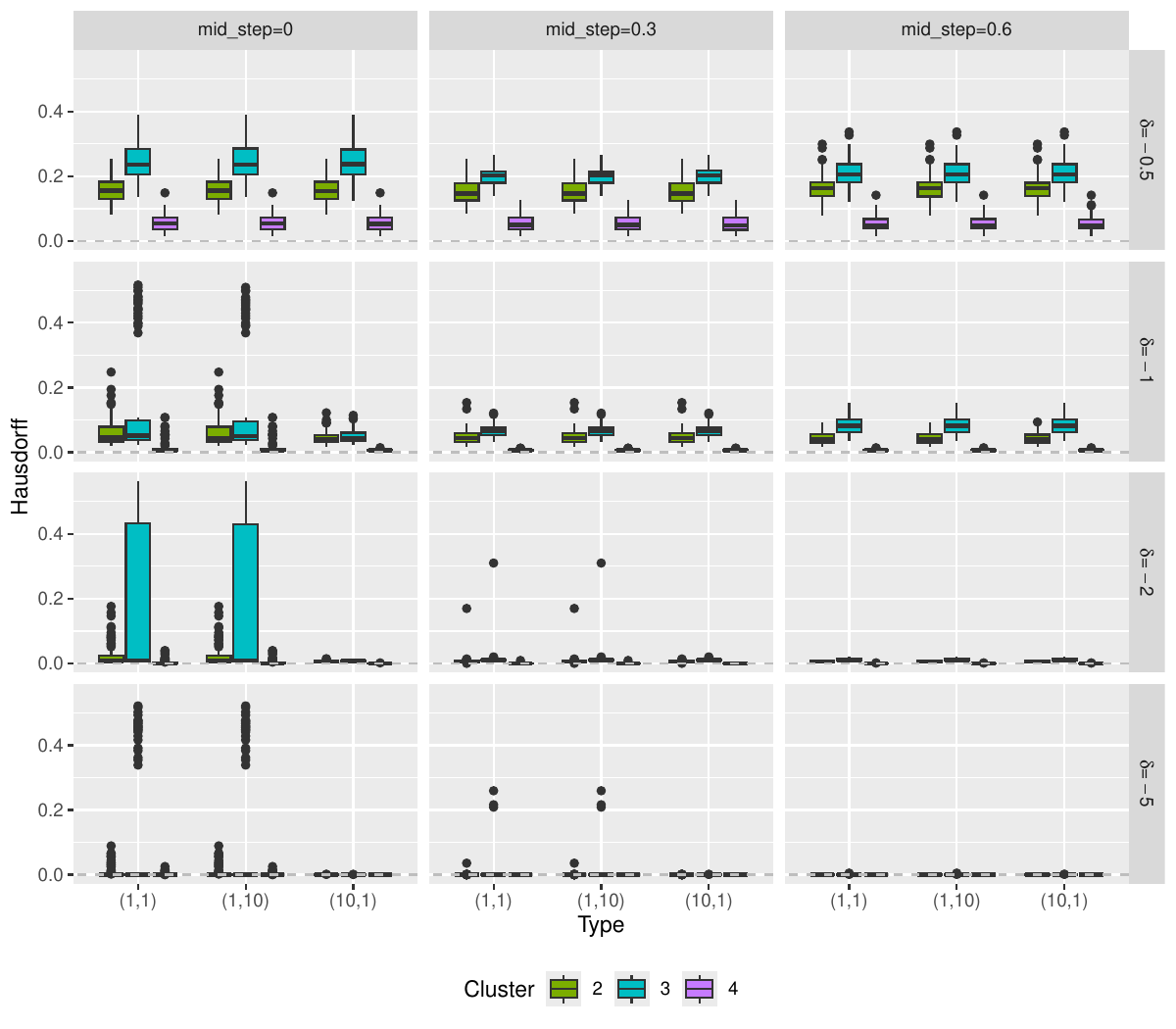}
    \caption{Boxplot of the normalized infinite distance $d\left(\bold{t}^{\left(k\right)},\widehat{\bold{t}}^{\left(k\right)}\right)$ separated according to the true cluster (colors) for the different values of $\delta$ (lines), the different length between the two change-points in cluster $3$ (columns) and different values for the couple type=$(Nb_{\text{init}},Nb_{\text{M-step}})$ where $Nb_{\text{init}}$ is the number of initializations and $Nb_{\text{M-step}}$ is the number of iterations in the M-step (x-axis).}
    \label{fig:Hausdorff}
\end{figure}

\noindent {\bf{Misclassification.}} Figure \ref{fig:Erreur} represents the classification error for the different studied configurations. It shows that the error rate increases as the jump size decreases. As can be see in Figure \ref{fig:Origine_destination}, which represents the proportions of cluster membership estimates as a function of the true cluster, only cluster $3$ is correctly estimated. For profiles in other clusters, if the detection problem is difficult ($\delta=-1$), all profiles (or almost all) are classified in cluster $3$, and if the detection problem is a little less difficult ($\delta=-2$), profiles from cluster $4$ are essentially classified in clusters $3$ and $4$ (this is more marked when mid\_step= 0, i.e. the two change-points of profiles in cluster $3$, are very closed), profiles from cluster $2$ in clusters $2$ and $3$, and those from cluster $1$ in the first three clusters. When the detection problem is easy ($\delta=-5$), all the profiles are correctly classified. In Appendix \ref{sec:P34}, we show that, theoretically, a profile originating from cluster $4$ has a non-negligible probability of belonging to cluster $3$, especially for small jumps (i.e., small $\delta$), particularly when the resulting segmentation in the cluster has two closely spaced change-points (see Figure \ref{fig:prob}). \\

\begin{figure}
    \centering
    \includegraphics[width=\linewidth]{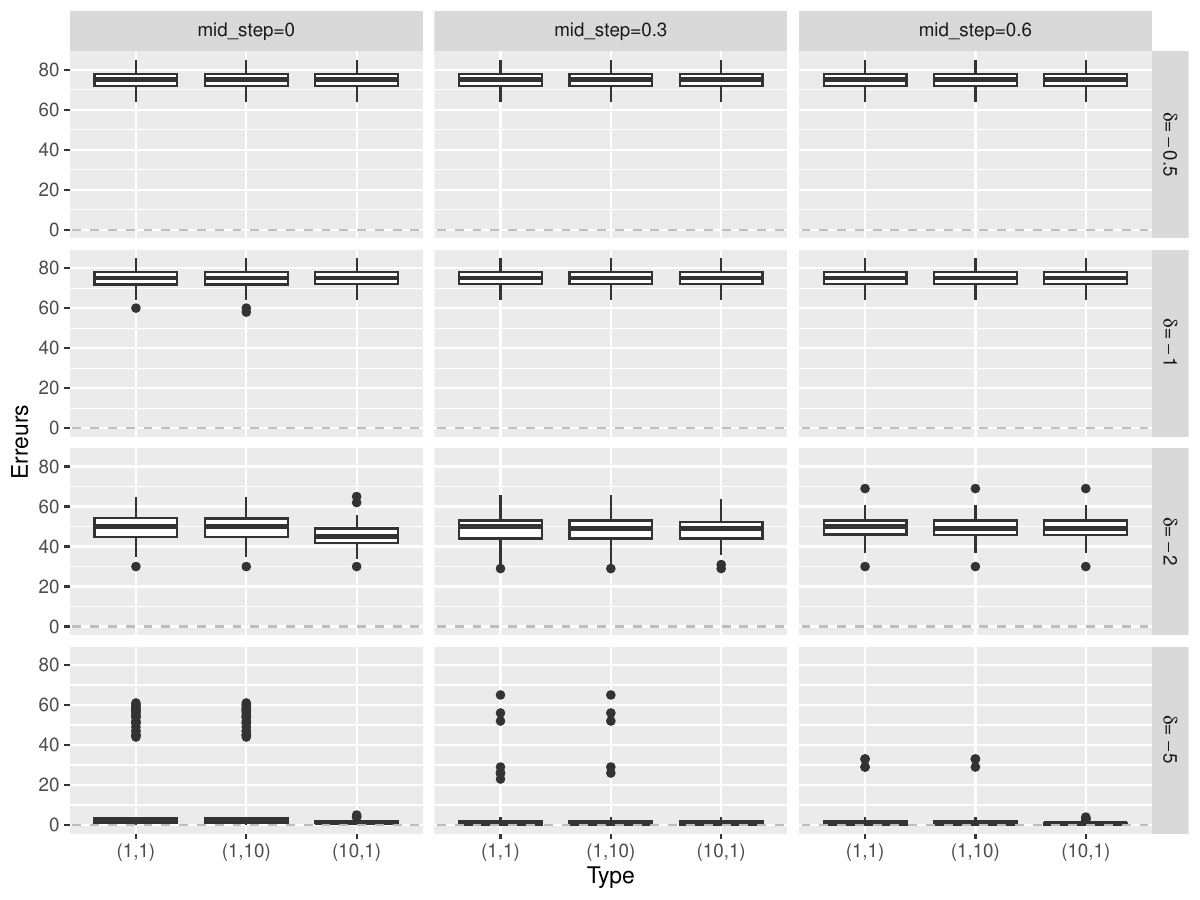}
    \caption{Boxplot of the misclassification rate for the different values of $\delta$ (lines), the different length between the two change-points in cluster $3$ (columns) and different values for the couple type=$(Nb_{\text{init}},Nb_{\text{M-step}})$ where $Nb_{\text{init}}$ is the number of initializations and $Nb_{\text{M-step}}$ is the number of iterations in the M-step (x-axis).}
    \label{fig:Erreur}
\end{figure}

\begin{figure}
    \centering
    \includegraphics[width=\linewidth]{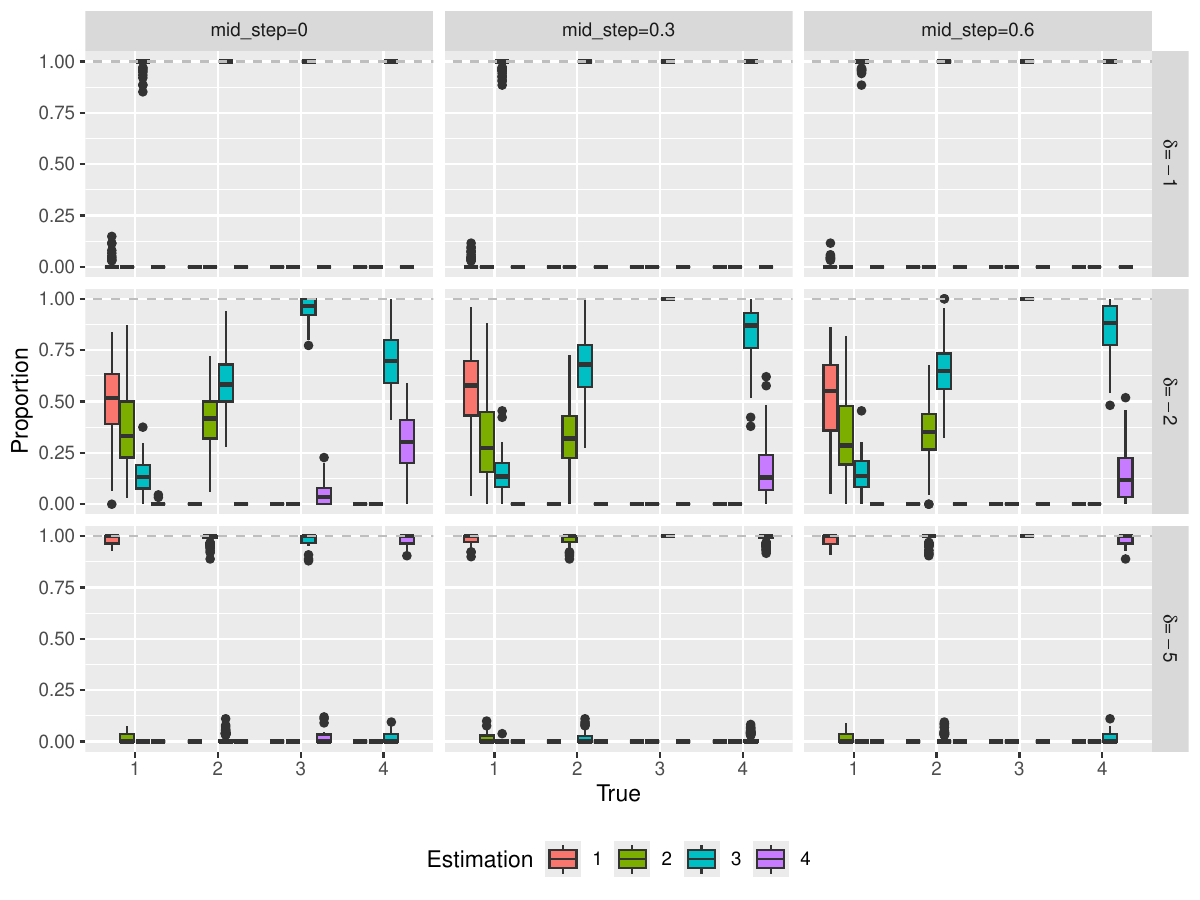}
    \caption{Boxplot of the proportions of cluster membership estimates (color) as a function of the original cluster (x-axis) for three values of $\delta$ (lines) and the different length between the two change-points in cluster $3$ (rows). Only for the case with $Nb_{\text{init}}=10$ initializations and $Nb_{\text{M-step}}=1$ iteration in the M-step.}
    \label{fig:Origine_destination}
\end{figure}

\noindent {\bf{Number of iterations.}} All results improve with the number of initializations $(Nb_{\text{init}})$ as expected, but not with the number of iterations within the M-step $(Nb_{\text{M-step}})$. In the two following studies, these numbers are fixed to $Nb_{\text{init}}=10$ and $Nb_{\text{M-step}}=1$.

\subsubsection{Study 2.} Figures~\ref{fig:asymp:Delta_norm}, \ref{fig:asymp:Hausdorff}, \ref{fig:asymp:Erreur}, and \ref{fig:asymp:Origine_destination} show that increasing the number of profiles reduces the variance of the estimates, namely $\delta$, the locations of the change-points and the membership of the clusters but does not necessarily improve their accuracy. This is what we also observed when examining more precisely the cluster membership as a function of the true cluster of origin (see Figure~\ref{fig:asymp:Origine_destination}). 

However, as expected and typical in change-point detection frameworks, the location of the change-points improves with the increase of profile length. Surprisingly, though, this does not enhance either the estimation of $\delta$ or its accuracy, nor does it improve the clustering assignment. The estimation of $\delta$ is not consistent with the calculation of the precision estimate (see Section \ref{sec:Fisher}), which states that the estimator becomes more accurate as the number of data increases.

\begin{figure}
    \centering
    \includegraphics[width=\linewidth]{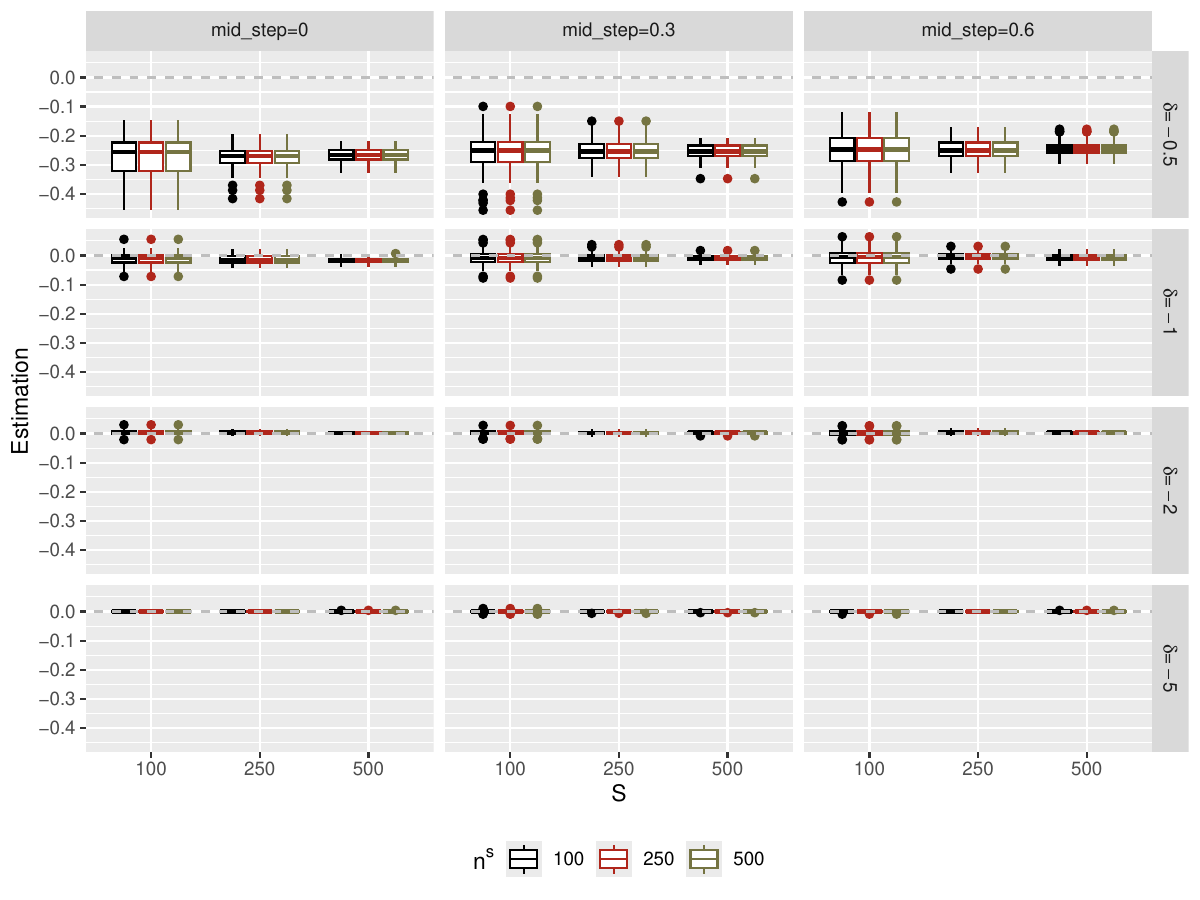}
    \caption{Boxplot of $d_r(\widehat{\delta},\delta)$ for the different values of $\delta$ (lines), the different length between the two change-points in cluster $3$ (columns), the number $S$ of profiles (on the x-axis) and the number $\ns$ of observations (box outline color).}
    \label{fig:asymp:Delta_norm}
\end{figure}

\begin{figure}
    \centering
    \includegraphics[width=\linewidth]{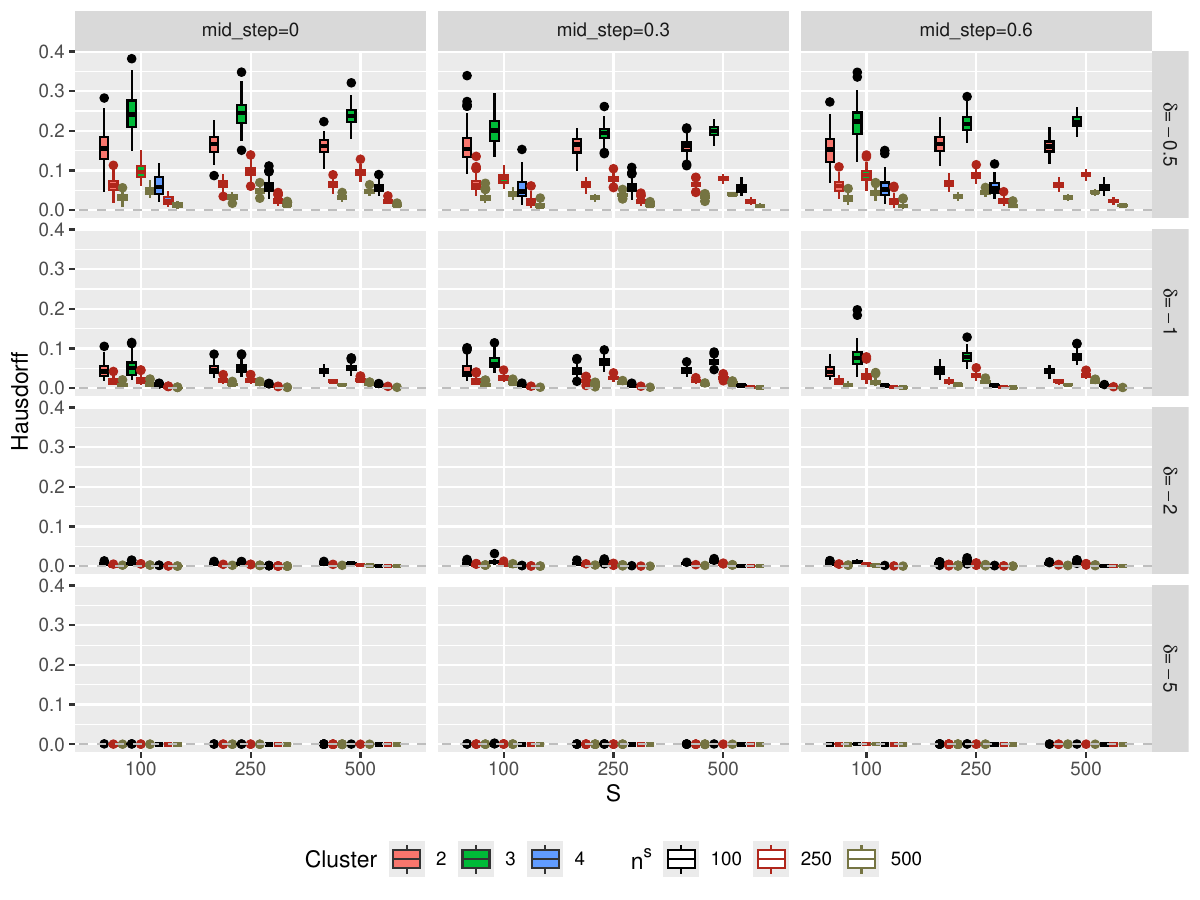}
    \caption{Boxplot of the normalized infinite distance $d\left(\bold{t}^{\left(k\right)},\widehat{\bold{t}}^{\left(k\right)}\right)$ separated according to the true cluster (colors) for the different values of $\delta$ (lines), the different length between the two change-points in cluster $3$ (columns), the number $S$ of profiles (on the x-axis) and the number $\ns$ of observations (box outline color). }
    \label{fig:asymp:Hausdorff}
\end{figure}

\begin{figure}
    \centering
    \includegraphics[width=\linewidth]{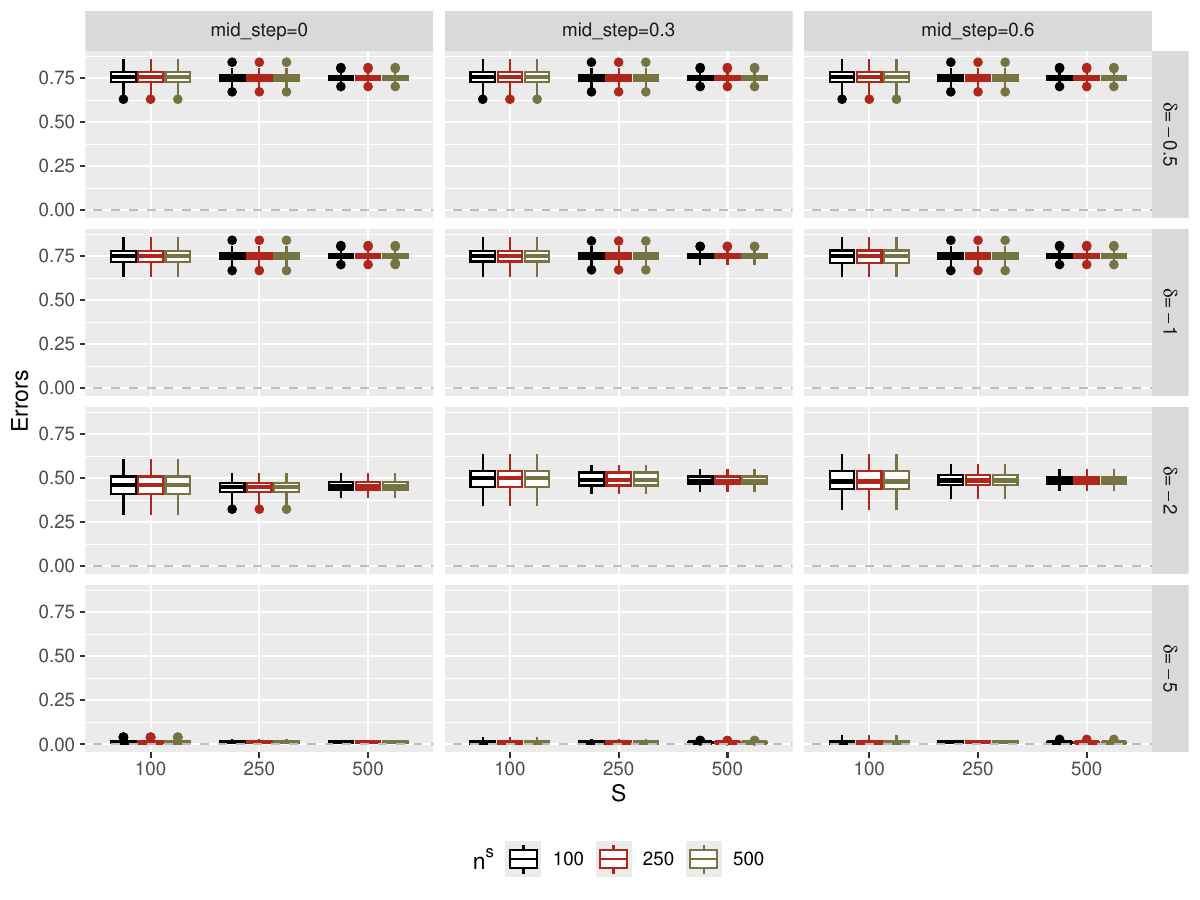}
    \caption{Boxplot of the misclassification rate for the different length $\ns$ for the different values of $\delta$ (lines), the different length between the two change-points in cluster $3$ (columns), the number $S$ of profiles (on the x-axis) and the number $\ns$ of observations (box outline color).}
    \label{fig:asymp:Erreur}
\end{figure}

\begin{figure}
    \centering
    \begin{tabular}{c}
    $S=n^s=100$\\
    \includegraphics[width=\linewidth]{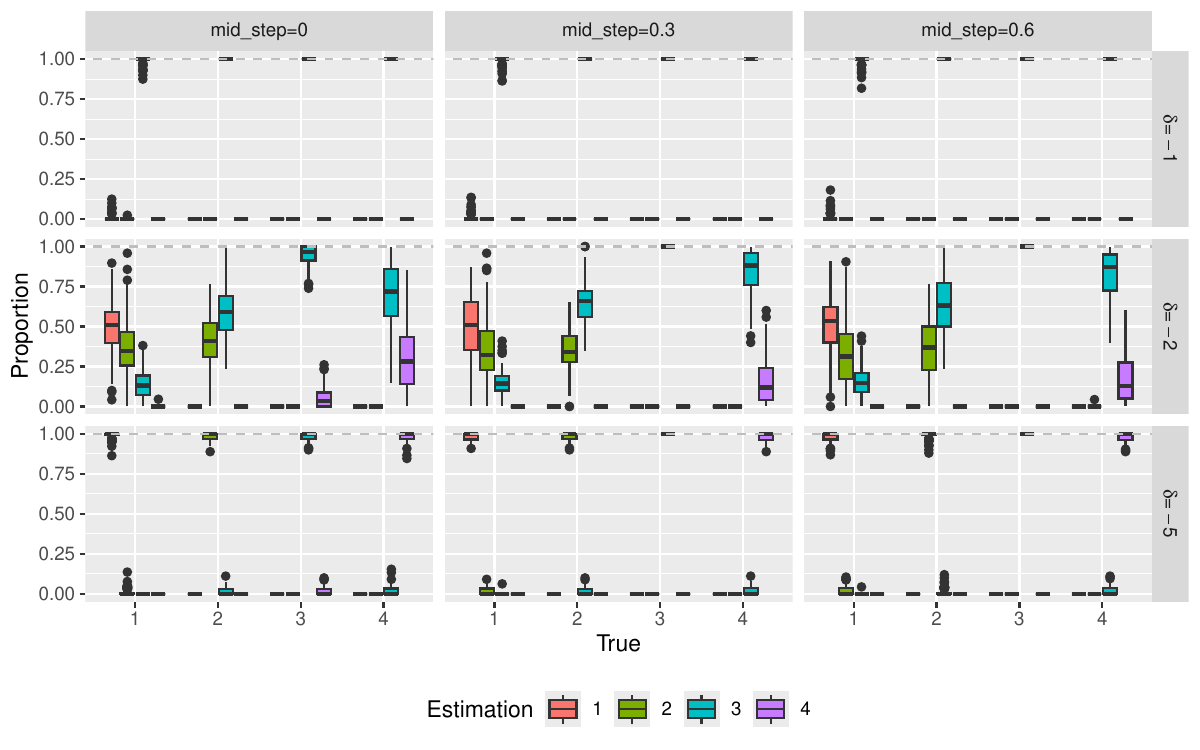}\\
    \hline
    \\
    $S=n^s=500$\\
    \includegraphics[width=\linewidth]{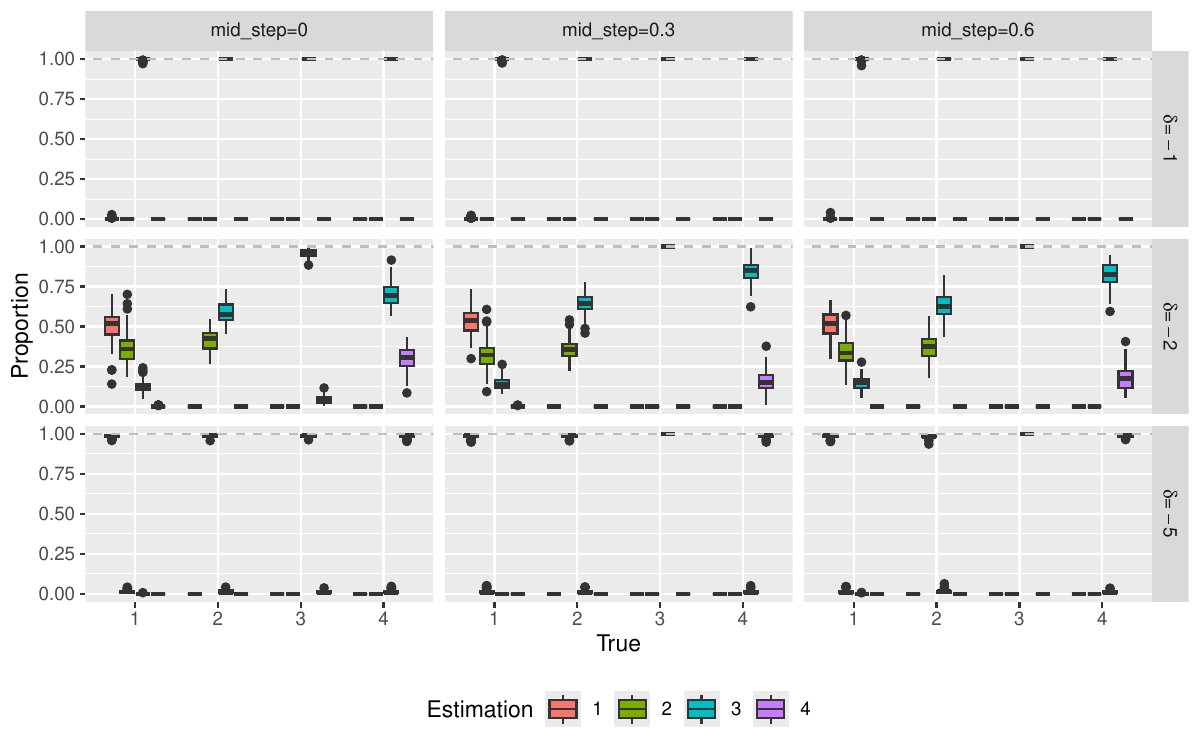}\\
    \\
    \end{tabular}
    \caption{Boxplot of the proportions of cluster membership estimates (color) as a function of the original cluster (x-axis) for three values of $\delta$ (lines) and the different length between the two change-points in cluster $3$ (columns) with $S=n^s=100$ (top graphic) and $S=n^s=500$ (bottom graphic).}
    \label{fig:asymp:Origine_destination}
\end{figure}

\subsubsection{Study 3.} In Figure~\ref{fig:delta:Erreur}, several insights can be drawn regarding the impact of overestimating the parameter $\delta$. When $\delta$ is very small - meaning the change-point detection task is inherently difficult - strongly overestimating its value (e.g., using 1.5$\delta$ or 2$\delta$) actually reduces (very slightly) the number of errors. We need to overestimate a lot in order to compensate the segmentation with two change-points which lead inevitably to the best fit. %This suggests that a looser threshold may help avoid missing true change-points. 
As $\delta$ increases (the detection problem becomes easier), a slight overestimation (by a factor of approximately 1.1 to 1.2) performs better than using the true value. This improvement is likely due to the fact that clusters $2$ and $4$ become more easily distinguishable from cluster $3$, leading to more accurate segmentation. However, when $\delta$ is overestimated too much in this case (i.e., by a factor of 1.5 or 2 in the case where the true $\delta$ value is high), the performance deteriorates. %This is because overly lenient thresholds introduce spurious detections, which degrade both change-point estimation and clustering quality. 
Additionally, the effect of overestimating $\delta$ is influenced by the length of the central segment in cluster 3. This segment length only significantly affects the results when $\delta$ is overestimated by a factor of 1.5. Indeed, if the threshold is overestimated (e.g., set to 1.5$\delta$), the change-points become harder to detect. In this case, a longer central segment helps to better distinguish segments and facilitates detection, thereby improving classification.

\begin{figure}
    \centering
    \includegraphics[width=\linewidth]{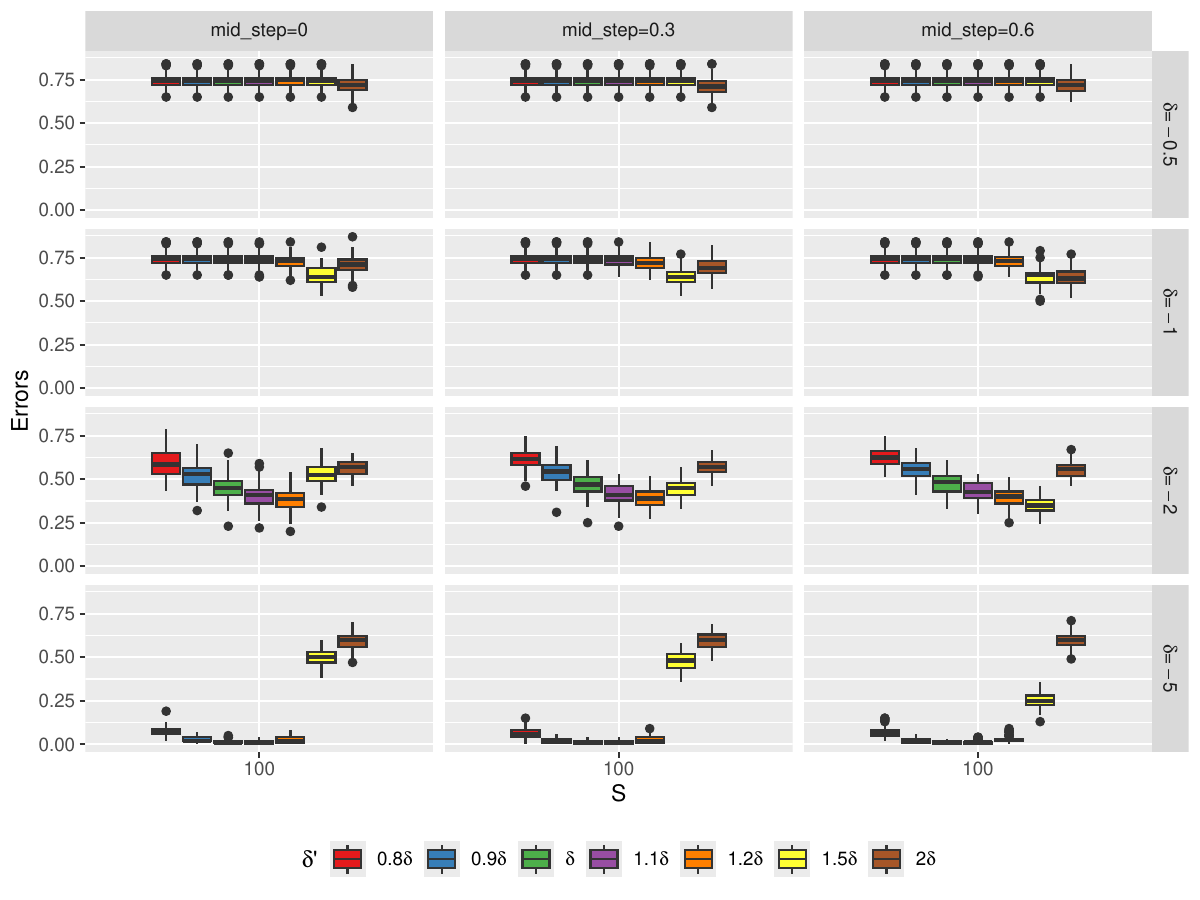}
    \caption{Boxplot of the misclassification rate for the different values of the true value $\delta$ (lines), the different length between the two change-points in cluster $3$ (columns) and the value of $\delta'\in\{0.8\delta,0.9\delta,\delta,1.1\delta,1.2\delta,1.5\delta,2\delta\}$ chosen (colors). For these simulations, the number of profiles is $S=100$ for $n^s=250$ observations by profile.}
    \label{fig:delta:Erreur}
\end{figure}

\FloatBarrier

\section{Application} \label{sec:application}

\noindent {\bf{Dataset.}} Data were collected from Single Molecule Photo Bleaching (SMPB) experiments. Fluorescent proteins were adsorbed onto a plasma-cleaned glass coverslip at an extremely low concentration (approximately $10$ to $50$ pM), allowing them to appear as individual spots. These spots were continuously illuminated with a high-power green laser ($488$ nm), causing their fluorescence to bleach over time. The fluorescence intensity of each spot was recorded continuously at a rate of $20$ images per second for $15$ to $30$ seconds, using a Total Internal Reflection Fluorescence (TIRF) microscope. The detailed experimental procedure is described in \citep{stoppin2020studying}. Two fluorescent proteins were used in the SMPB experiments: sfGFP (superfolder Green Fluorescent Protein \citep{pedelacq2006engineering}), a monomeric protein expected to bleach in a single step, and sfGFP-EB3, a dimeric fusion protein \citep{de2010molecular} anticipated to bleach in two distinct steps. In total, $120$ fluorescence intensity profiles were obtained, $85$ for sfGFP and $35$ for sfGFP-EB3. \\

\noindent {\bf{Results.}} The distribution of the cluster membership is  represented in Figure \ref{fig:enter-label}
separated according the type of the protein. As expected, sfGFP proteins are mainly classified in cluster $2$ and sfGFP-EB3 proteins in cluster $3$. None is classified in cluster $4$. The estimated value of $\delta$ is $-29.4$. The mean estimate of the standard deviation $\overline{sigma_s}$ on all profiles is 11.7 (with a variability of 13.2); the closest case in our simulations is $\delta=-2$, i.e. a case of moderately easy detection.

After verification, it appears that the apparent misclassification cases (5 sfGFP in cluster $3$ instead of cluster $2$; 9 sfGFP-EB3 in cluster $2$ instead of cluster $3$) are due to biological heterogeneity of the samples (sfGFP proteins that make dimers; misfolded sfGFP-EB3 proteins that fail to form a dimer). So our segmentation procedure revealed very efficient in determining the oligomerization state of the proteins. We also observed the segmentation is efficient in segmenting challenging traces with a $\delta$ value in the same order of magnitude than the noise of the profile (see Figure \ref{fig:segmentation_free_constrained} for example).

\begin{table}%[h]
 \centering
 \caption{Cluster membership of the sfGFP and sfGFP-EB3 profiles.}
 \label{tab:Appli_Cluster_Member}
 \begin{tabular}{l|cccc|cccc}
   & \multicolumn{4}{@{}c@{}}{sfGFP} & \multicolumn{4}{@{}c@{}}{sfGFP-EB3} \\
    cluster number &  1 & 2 & 3 & 4 &  1 & 2 & 3 & 4 \\
   \hline     \hline
    count  & 5 & 75 & 5 & 0 &0& 9 & 26 & 0 \\
 \end{tabular}
\end{table}

\begin{figure}
\centering
\includegraphics[width=\linewidth]{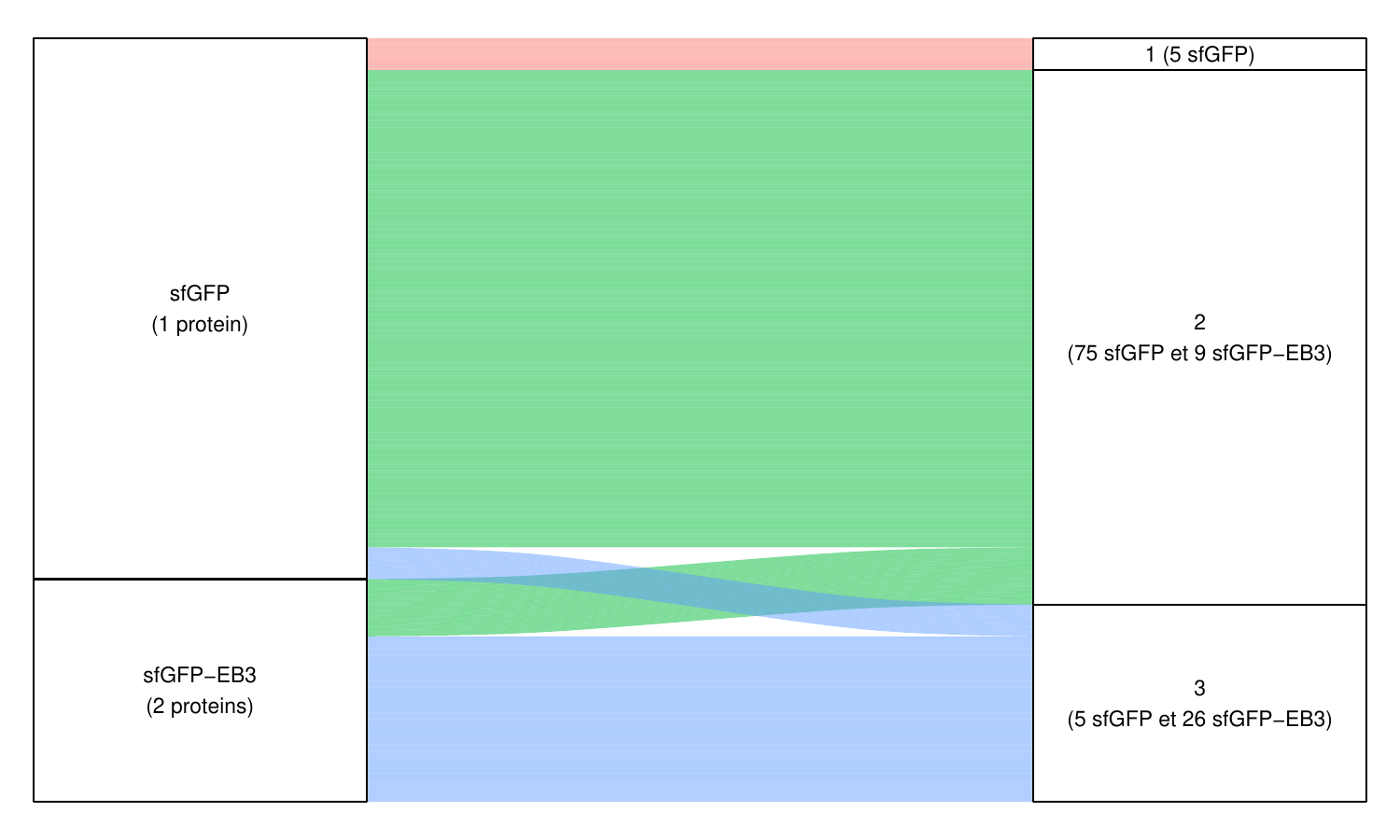}
    \caption{Cluster membership of the sfGFP and sfGFP-EB3 profiles where the colors are associated to the estimated cluster (salmon for the cluster $1$, green for $2$ and blue for $3$).}
    \label{fig:enter-label}
\end{figure}

\section{Discussion} \label{sec:discussion}

\noindent {\bf{Classification procedure for multiple profiles.}} We proposed a unified classification model for multiple profiles, based on their segmented profiles which are linked through the size of the jumps. This work was motivated by an application in neuroscience, in which the number of clusters and their profile structures were predefined: four scenarios were considered, corresponding to a profile with no change-point, a profile with a single change-point and a jump of size $\delta$, a profile with two change-points, each with a jump of size $\delta$, and a profile with one change-point and a jump of size $2\delta$. The number of clusters and the number of change-points were therefore not in this study treated as parameters to choose. The procedure involves a EM algorithm for parameter inference, during which combined with an exhaustive search to locate the change-points is performed. The classical dynamic programming (DP) algorithm could not be used in this context, as the segment-additivity condition required for its application is not satisfied here. Our results showed that cluster $3$, corresponding to profiles with two change-points, was the most attractive (i.e., most frequently selected by the algorithm). \\

\noindent {\bf{Possible Improvements.}} The procedure can be improved in two directions:
\begin{itemize}
\item From an algorithmic perspective, computational efficiency could likely be improved by adopting more recent dynamic programming approaches, such as the gfpop algorithm proposed by \cite{runge2023gfpop}. This algorithm is not only faster but also allows for the integration of constraints between successive segments, which aligns well with the requirements of our setting.
\item From a statistical perspective, our model does not infer the number of change-points; instead, each profile is assumed to contain either $0$, $1$, or $2$ change-points. Naturally, increasing the number of change-points improves the fit to the data, which may cause profiles with a single change-point to be misclassified into the cluster characterized by two closely spaced change-points (cluster $3$). To mitigate this issue, one could introduce a penalized criterion that accounts for both the number and the length of segments, such as the multiscale penalty proposed by \cite{verzelen2023optimal}.
\end{itemize}

%%%%%%%%%%%%%%
\newpage

\begin{appendices}

\section{Understand the misclassification from cluster $4$ to cluster $3$} \label{sec:P34}

In this section, we study the probability of misclassifying a profile in cluster $3$ instead of its true cluster $4$, as a function of the jump $\delta$ and the segment lengths. First, let us introduce some notations and assumptions. \\

Let $Y^\star$ be a profile with length $n$ that belongs to cluster $4$. Its true distribution parameters are indexed by $^\star$: 
$$
Y^\star_t \sim \mathcal{N}(m_4^\star(t),\sigma^{\star,2}_4)
$$
where $m_4(t)^\star=\mu_{4}^\star+2  \  \delta \ \mathbb 1_{t>t_{41}^\star}$ or $m_{4}^\star =\mu_4^\star +T_4^\star \delta$ with $T_4^\star=\begin{bmatrix}  0_{n_{41}^\star} \\ 2 \mathbb 1_{n_{42}^\star} \end{bmatrix}$. We note $\epsilon^\star_t=Y^{\star}_t-m_4^\star(t)=Y^{\star}_t-(T_4^\star)_t \delta-\mu_4^\star$. 

The segmentation in cluster $3$ is given by $\llbracket 1, t_{31} \rrbracket \cup \llbracket t_{31}+1, t_{32} \rrbracket \cup \llbracket t_{32}+1, n \rrbracket$, and is denoted by $T_3$ for simplicity. The true segmentation is $\llbracket 1, t_{41}^\star \rrbracket \cup \llbracket t_{41}^\star+1, n \rrbracket$, and is denoted by $T_4^\star$. We assume that
\begin{itemize}
    \item the segmentation in cluster $4$ is the true one and the unique change-point of $T_4^\star$ belongs to the second segment of $T_3$: $t_{31} < t_{41}^\star < t_{32}$,
    \item the proportions of the two clusters are known and equal, i.e., $\pi_3 = \pi_4$.
\end{itemize}

\bigskip 

The profile $Y^\star$ will be classified in cluster $3$ instead of cluster $4$ if 
\begin{eqnarray*}
&& \tau_3^\star = \mathbb{P}(Z^\star_3=1 \vert Y^\star) > \tau_4^\star = \mathbb{P}(Z^\star_4=1 \vert Y^\star)  \Leftrightarrow  \frac{\tau_3^\star}{\tau_4^\star} >1, \\
\end{eqnarray*}
or again if
\begin{eqnarray*}
&&\log{\left (\frac{\tau_3^\star}{\tau_4^\star}\right )} >0  \Leftrightarrow  \log{\left (\frac{\pi_3}{\pi_4}\right )}+\log{\left (\frac{f(Y^\star;\delta,T_3,\widehat{\mu}_3,\widehat{\sigma}^2_3}{f(Y^\star;\delta,T_4^\star,\widehat{\mu}_4,\widehat{\sigma}^2_4}\right )} >0, \\
\end{eqnarray*}
or again, by assumption, if
\begin{eqnarray*}
&&\log{\left (\frac{f(Y^\star;\delta,T_3,\widehat{\mu}_3,\widehat{\sigma}^2_3}{f(Y^\star;\delta,T_4^\star,\widehat{\mu}_4,\widehat{\sigma}^2_4}\right )} >0  \Leftrightarrow Q_n(T_3)-Q_n(T_4^\star)<0,\\
\end{eqnarray*}
where 
$$
Q_n(T_k)=Q_n(Y^\star;T_k,\widehat{\mu}_k,\delta)= \sum_{t=1}^n (Y^\star_t-\widehat{m}_k(t))^2=  \sum_{t=1}^n (Y^\star_t-\widehat{\mu}_k(t)-(T_k)_t \delta)^2,
$$
with we recall that $\widehat{\mu}_k:=\widehat{\mu}_k(T_k)=\frac{1}{n}\sum_{t=1}^n (Y^{\star}_t-(T_k)_t \delta)$.

\bigskip

Proposition \ref{prop_princ} provides the probability of misclassifying a profile from the true cluster $4$ as belonging to the cluster $3$. As expected, this probability increases as
\begin{itemize}
    \item the signal-to-noise ratio at the change position $|\delta|/\sigma_4^\star$ decreases,
    \item when the  segmentations $T_3$ and $T_4^\star$ are closed (the value of $V$, given in the proposition, decreases).
\end{itemize}
This reflects the idea that recovering the true cluster, or distinguishing between the two clusters, becomes more difficult when either the change-point detection itself is challenging (small jump compared to the variability) and/or the true change-point is close to the two change-points of segmentation $T_3$. This point is more marked for the size of the jump. Figure \ref{fig:prob} shows the behavior of this probability as a function of $n_3 = n_2 = n_m$ and for different values of $\delta$ when $\sigma^\star_4 = 1$. For large $|\delta|$ (about $-3$) and small $n_m$, the probability to classify the profile in cluster $3$ is not null.

\begin{center}
\begin{figure}
    \centering
    \includegraphics[scale=0.25]{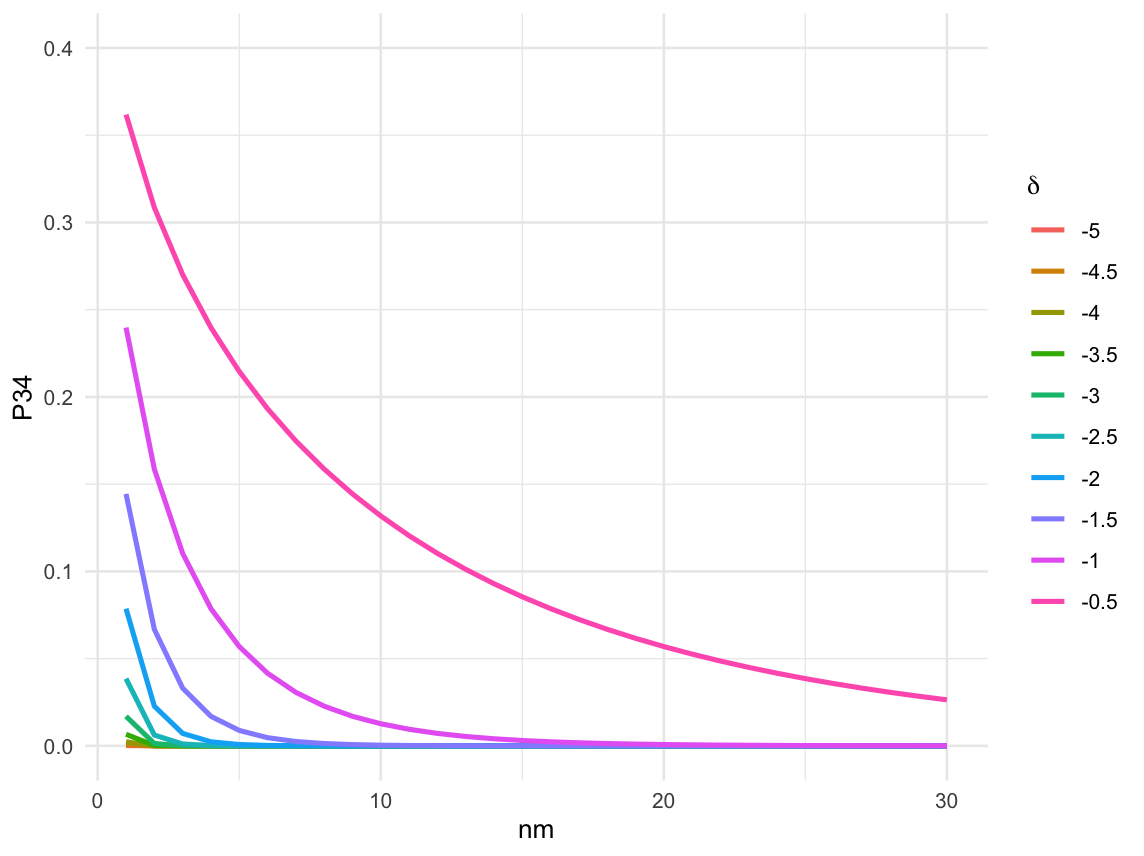}
    \caption{Probability of misclassification given by \eqref{eq:proba34} as a function of $n_2=n_3=n_m$ and for different values of $\delta$.}
    \label{fig:prob}
\end{figure}
\end{center}

\begin{proposition}\label{prop_princ}
We have that 
\begin{equation} \label{eq:proba34}
\mathbb{P}(Q_n(T_3)-Q_n(T_4^\star)<0)=1-F \left (-\frac{\delta \sqrt{V}}{2 \sigma^\star_4} \right),
\end{equation}
where $F$ is the cumulative function of $\mathcal{N}(0,1)$ and $V=n_2+n_3-\frac{(n_2-n_3)^2}{n}>0$ with $n_2=t_{41}^\star-t_{31}+1$ and $n_3=t_{32}-t_{41}^\star+1$.
\end{proposition}

\begin{proof}
We consider the segmentation of the grid $\llbracket 1, n \rrbracket$ resulting from the intersection of the segmentations $T_3$ and $T_4^\star$, that is $\bigcup_i J_i$ where
\begin{itemize}
\item $J_1=\llbracket 0, t_{31}\rrbracket$ with length $n_1$, 
\item $J_2=\llbracket t_{31}+1,t_{41}^\star \rrbracket$ with length $n_2$,
\item $J_3=\llbracket t_{41}^\star+1, t_{32}\rrbracket$ with length $n_3$, 
\item $J_4=\llbracket t_{32}+1,n\rrbracket$ with length $n_4$.
\end{itemize}
We define the two following profiles which represent the profile $Y^\star$ centered according to segmentations $T_3$ and $T_4^\star$: 
$$
\widetilde{Y}^{\star,3}=Y^{\star}-T_3 \delta, \ \ \ \  \text{and} \ \ \ \  \widetilde{Y}^{\star,4}=Y^{\star}-T_4^\star \delta
$$
On each segment $J_i$, we get
\begin{itemize}
\item on $J_1$, $\widetilde{Y}^{\star,3}_t={Y}^{\star}_t=\widetilde{Y}^{\star,4}_t$, 
\item on $J_2$, $\widetilde{Y}^{\star,3}_t={Y}^{\star}_t-\delta=\widetilde{Y}^{\star,4}_t-\delta$,
\item on $J_3$, $\widetilde{Y}^{\star,3}_t={Y}^{\star}_t-\delta=\widetilde{Y}^{\star,4}_t+\delta$,
\item on $J_4$, $\widetilde{Y}^{\star,3}_t={Y}^{\star}_t-2 \delta=\widetilde{Y}^{\star,4}_t $.
\end{itemize}
A simulated profile ${Y}^{\star}$ with its associated centered profiles $\widetilde{Y}^{\star,3}$ and $\widetilde{Y}^{\star,4}$ are represented in Figure \ref{fig:Datas}.
\begin{figure}[!ht]
    \centering
    \includegraphics[width=\linewidth]{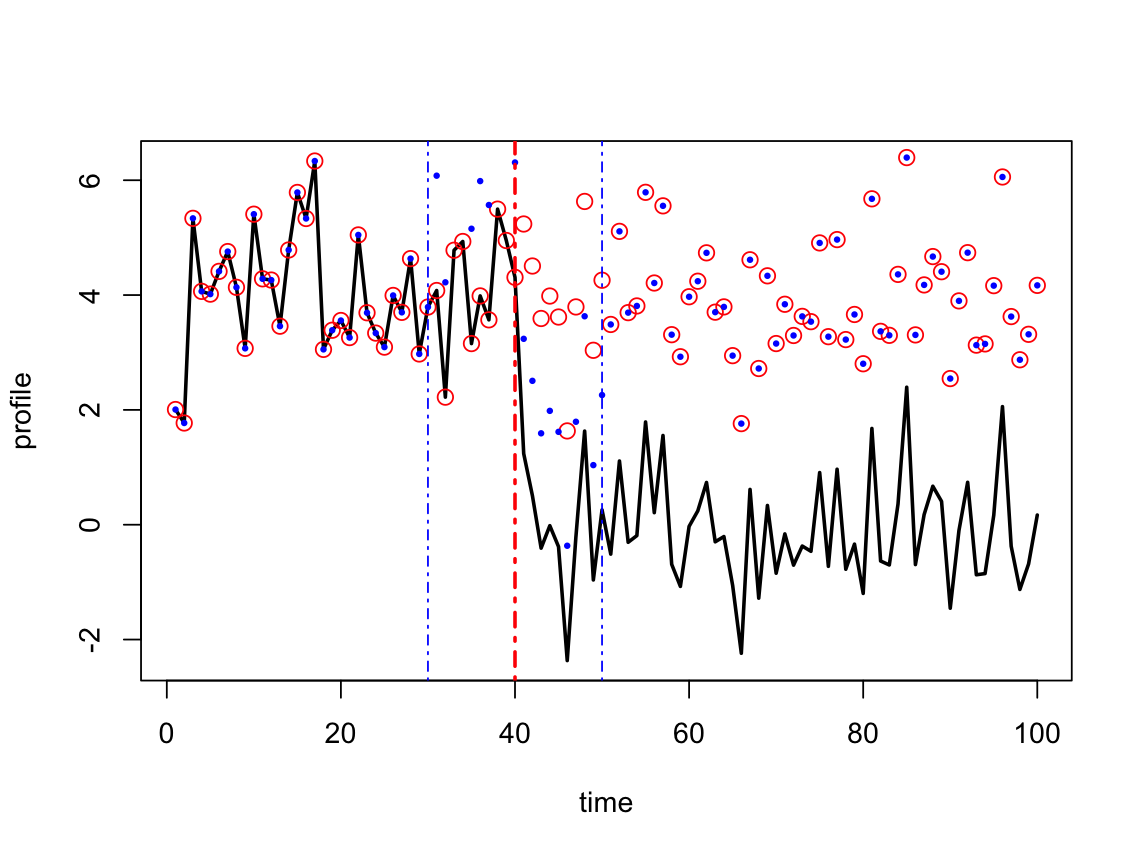}
    \caption{In black: ${Y}^{\star}$, in blue: $\widetilde{Y}^{\star,3}$ and in red: $\widetilde{Y}^{\star,4}$. The dotted vertical lines are the change-points $t_{31}$ and $t_{32}$ in blue and $t_{41}^\star$ in red.}
    \label{fig:Datas}
\end{figure}
We have that  
\begin{eqnarray*}
Q_n(T_3)&=& \sum_{t=1}^n (Y^\star_t-\widehat{m}_3(t))^2= \sum_{t=1}^n (\widetilde{Y}^{\star,3}_t-\widehat{\mu}_3)^2 \\
&=& \sum_{t \in J_1 \cup J_4} (\widetilde{Y}^{\star,4}_t-\widehat{\mu}_3)^2+\sum_{t \in J_2} (\widetilde{Y}^{\star,4}_t-\delta-\widehat{\mu}_3)^2+\sum_{t \in J_3} (\widetilde{Y}^{\star,4}_t+\delta-\widehat{\mu}_3)^2 \\
%&=& \sum_{t \in J_1 \cup J_4} (\epsilon_t^\star+\mu_4^\star-\widehat{\mu}_3)^2+\sum_{t \in J_2} (\epsilon_t^\star+\mu_4^\star-\delta-\widehat{\mu}_3)^2+\sum_{t \in J_3} (\epsilon_t^\star+\mu_4^\star+\delta-\widehat{\mu}_3)^2 \\
&=& \sum_{t=1}^n (\widetilde{Y}^{\star,4}_t-\widehat{\mu}_3)^2+\delta^2 (n_2+n_3)+2\delta \left ( \sum_{t \in J_3} (\widetilde{Y}^{\star,4}_t-\widehat{\mu}_3)- \sum_{t \in J_2} (\widetilde{Y}^{\star,4}_t-\widehat{\mu}_3)  \right ).
\end{eqnarray*}
Since $\widehat{\mu}_3=\widehat{\mu}_4-\frac{\delta}{n} (n_2-n_3)$ and $\widetilde{Y}^{\star,4}_t=\mu_4^\star+\epsilon^\star_t$, we get
\begin{eqnarray*}
Q_n(T_3)&=& 
%Q_n(T_4^\star)+\sum_{t=1}^n (\widetilde{Y}^{\star,4}_t-\widehat{\mu}_3)^2- \sum_{t=1}^n (\widetilde{Y}^{\star,4}_t-\widehat{\mu}_4)^2 \\
%&& +\delta^2 (n_2+n_3)+2\delta \left ( \sum_{t \in J_3} (\widetilde{Y}^{\star,4}_t-\widehat{\mu}_3)- \sum_{t \in J_2} (\widetilde{Y}^{\star,4}_t-\widehat{\mu}_3)  \right ) \\
%&=&
Q_n(T_4^\star)+2 \frac{\delta}{n} (n_2-n_3) \sum_{t=1}^n (\widetilde{Y}^{\star,4}_t-\widehat{\mu}_4)+\frac{\delta^2}{n} (n_2-n_3)^2+\delta^2 (n_2+n_3)\\
&& +2\delta \left ( \sum_{t \in J_3} (\widetilde{Y}^{\star,4}_t-\widehat{\mu}_4)- \sum_{t \in J_2} (\widetilde{Y}^{\star,4}_t-\widehat{\mu}_4)  \right )+2\frac{\delta^2}{n} (n_2-n_3) (n_3-n_2) \\
&=& Q_n(T_4^\star)+2 \delta \left (  \frac{(n_2-n_3)}{n}\sum_{t=1}^n (\widetilde{Y}^{\star,4}_t-\widehat{\mu}_4) + \sum_{t \in J_3} (\widetilde{Y}^{\star,4}_t-\widehat{\mu}_4)- \sum_{t \in J_2} (\widetilde{Y}^{\star,4}_t-\widehat{\mu}_4)\right ) \\
&& + \delta^2 \left (n_2+n_3-\frac{(n_2-n_3)^2}{n} \right ) \\
&=& Q_n(T_4^\star)+ \delta^2 \left (n_2+n_3-\frac{(n_2-n_3)^2}{n} \right )\\
&& +2 \delta \left (  \frac{(n_2-n_3)}{n} \sum_{t \in J_1 \cup J_4} \widetilde{Y}^{\star,4}_t + \left ( \frac{(n_2-n_3)}{n} +1\right )\sum_{t \in J_3} \widetilde{Y}^{\star,4}_t+ \left ( \frac{(n_2-n_3)}{n} -1\right )\sum_{t \in J_2} \widetilde{Y}^{\star,4}_t \right ) \\
&=& Q_n(T_4^\star)+ \delta^2 \left (n_2+n_3-\frac{(n_2-n_3)^2}{n} \right )\\
&&+2 \delta \left (\frac{(n_2-n_3)}{n} \sum_{t \in J_1 \cup J_4} \epsilon^\star + \left ( \frac{(n_2-n_3)}{n} +1\right )\sum_{t \in J_3}\epsilon^\star+ \left ( \frac{(n_2-n_3)}{n} -1\right )\sum_{t \in J_2} \epsilon^\star \right ). \\
\end{eqnarray*}
Finally, we have that
\begin{eqnarray*}
Q_n(T_3)-Q_n(T_4^\star)&=& 2 \delta A+ \delta^2 V,
\end{eqnarray*}
with
$$
A=\frac{(n_2-n_3)}{n}\sum_{J_1 \cup J_4} \epsilon_t^\star + \left ( \frac{(n_2-n_3)}{n} +1\right )\sum_{t \in J_3} \epsilon_t^\star+ \left ( \frac{(n_2-n_3)}{n} -1\right )\sum_{t \in J_2} \epsilon_t^\star,
$$
and
$$
V=n_2+n_3-\frac{(n_2-n_3)^2}{n}>0.
$$
Since $\epsilon_t^\star \sim \mathcal{N}(0,\sigma^{\star,2}_4)$, it is straightforward to see that
$$
A \sim \mathcal{N}(0,\sigma^{\star,2}_4 V),
$$
which concludes the proof. 
\end{proof}

\section{Calculation of the complete-data information matrix}  \label{sec:Calcul-Fisher}

In this section, we calculate $\mathcal{I}(\widehat{\theta})= \mathbb{E}_{X | Y} [I^c(\theta)]\big|_{\theta=\widehat{\theta}}$ where
\begin{equation*}
I^c(\theta)=- \frac{\partial^2 }{\partial \theta \partial {}^t \theta} \log{\mathbb{P}(Y,Z;\theta)},
\end{equation*}
and we recall that 
\begin{eqnarray*}
\log{\mathbb{P}(Y,Z;\theta)}&=&\sum_{s=1}^S \sum_{k=1}^4 Z_k^s \left ( \log(\pi_k)-\frac{n^s}{2} \log{(2 \pi \sigma_{k}^{2,s})}-\frac{1}{2\sigma_{k}^{2,s}} \| Y^s- \mu_{k}^{s} -T^s_k \delta\|_2^2 \right ).
\end{eqnarray*}

We first focus on the 'diagonal' term of the information matrix, thus for a second derivative of the complete-data information matrix w.r.t. each parameter: 

\begin{itemize}
\item proportion parameters $\pi_k$. Since $\sum_{k=1}^4 \pi_k=1$, it is sufficient to differentiate with respect to the first three components. For $k=1,2,3$, we get
\begin{eqnarray*}
\frac{\partial^2 }{\partial \pi_k^2} \log{\mathbb{P}(Y,Z;\theta)}&=& \frac{\partial^2 }{\partial \pi_k^2} \sum_{s=1}^S \left ( \sum_{k=1}^3  Z^s_k \log{\pi_k}+ Z^s_4 \log{(1-\sum_{k=1}^3 \pi_k)} \right)   = -\sum_{s=1}^S \left ( \frac{Z^s_k}{\pi_k^2}+\frac{Z^s_4}{\pi_4^2}\right ),
\end{eqnarray*}
thus it is straightforward to see that
\begin{eqnarray*}
\mathbb{E}_{X | Y} [I^c(\pi)] &=& \sum_{s=1}^S  
\begin{bmatrix}
\frac{\tau^s_1}{\pi_1^2}+\frac{\tau^s_4}{\pi_4^2} & \frac{\tau^s_4}{\pi_4^2} & \frac{\tau^s_4}{\pi_4^2} \\
\frac{\tau^s_4}{\pi_4^2} & \frac{\tau^s_2}{\pi_2^2}+\frac{\tau^s_4}{\pi_4^2} & \frac{\tau^s_4}{\pi_4^2} \\
\frac{\tau^s_4}{\pi_4^2} & \frac{\tau^s_4}{\pi_4^2} & \frac{\tau^s_3}{\pi_3^2}+\frac{\tau^s_4}{\pi_4^2},
\end{bmatrix}
\end{eqnarray*}
and since $\widehat{\pi}_k=\sum_{s} \tau^s_k/S$, we obtain
\begin{eqnarray*}
\mathbb{E}_{X | Y} [I^c(\pi)]\big|_{\pi=\widehat{\pi}} &=& S  
\begin{bmatrix}
\frac{1}{\widehat{\pi}_1}+\frac{1}{\widehat{\pi}_4} & \frac{1}{\widehat{\pi}_4} & \frac{1}{\widehat{\pi}_4} \\
\frac{1}{\widehat{\pi}_4}& \frac{1}{\widehat{\pi}_2}+\frac{1}{\widehat{\pi}_4} & \frac{1}{\widehat{\pi}_4} \\
\frac{1}{\widehat{\pi}_4}& \frac{1}{\widehat{\pi}_4} & \frac{1}{\widehat{\pi}_3}+\frac{1}{\widehat{\pi}_4},
\end{bmatrix}
\end{eqnarray*}

\item jump parameters $\delta$. To obtain the result, simply observe that
\begin{eqnarray*}
\frac{\partial^2 }{\partial \delta^2} \log{\mathbb{P}(Y,Z;\theta)}&=& - \sum_{s=1}^S \sum_{k=1}^4 \frac{ Z^s_k}{\widehat{\sigma}^{2,s}_k} \ {}^t {T_s^k} T_s^k.
\end{eqnarray*}

\item mean parameters $\mu_k^s$. Similarly,  
\begin{eqnarray*}
\frac{\partial^2 }{\partial (\mu_k^{s})^2} \log{\mathbb{P}(Y,Z;\theta)}&=& - \frac{n_s  Z^s_k}{{\sigma}^{2,s}_k} 
\end{eqnarray*}

\item variance parameters. We get 
\begin{eqnarray*}
\frac{\partial^2 }{\partial (\sigma_k^{2,s})^2} \log{\mathbb{P}(Y,Z;\theta)}&=&  \frac{Z^s_k}{2} \left ( \frac{n_s}{{\sigma}^{4,s}_k} - \frac{2}{{\sigma}^{6,s}_k} \| Y^s- \mu_{k}^{s} -T^s_k \delta\|_2^2\right). 
\end{eqnarray*}
Taking the conditional expectation at $\widehat{\theta}$, we obtain
\begin{eqnarray*}
\left ( \mathbb{E}_{X | Y} [I^c(\sigma^{2})]\big|_{\theta=\widehat{\theta}} \right )_{ks} &=&  \frac{n_s  \tau^s_k}{\widehat{\sigma}^{4,s}_k}. 
\end{eqnarray*}
\end{itemize}

For the extra-diagonal of the information matrix, we have the following:
\begin{itemize}
\item for the $(\mu,\sigma^2)$ block, 
\begin{eqnarray*}
 \frac{\partial^2 }{ \partial \mu_k^{s} \partial {\sigma}^{2,s}_k} \log{\mathbb{P}(Y,Z;\theta)}&=&  \frac{\partial }{\partial \mu_k^{s}}\left (   \frac{Z^s_k}{2 \sigma^{4,s}_k} \| Y^s- \mu_{k}^{s} -T^s_k \delta\|_2^2\right ) \\
&=&  \frac{Z^s_k}{2 \sigma^{4,s}_k} \left ( -2 \sum_{t=1}^{n_s} Y^s_t+2 \mu_k^{s} n_s +2 \ {}^t \mathbf{1} T_s^k \delta \right )
\end{eqnarray*}
By definition of $\widehat{\mu}_k^s$, we easily obtain that
\begin{eqnarray*}
 \mathbb{E}_{X | Y} [I^c(\mu,\sigma^2)]\big|_{\theta=\widehat{\theta}}  &=& 0 
\end{eqnarray*}
\item for the $(\delta,\mu)$ block, 
\begin{eqnarray*}
 \frac{\partial^2 }{\partial \delta \partial \mu_k^{s}} \log{\mathbb{P}(Y,Z;\theta)}&=&  \frac{\partial }{\partial \delta}  \left ( - \frac{Z^s_k}{2 \sigma^{2,s}_k}  \left ( -2 \ {}^t (Y^s_t) T_s^k +2 \mu_k^{s} n_s +2 \ {}^t \mathbf{1} T_s^k \delta \right ) \right ) \\
 &=& - \frac{Z^s_k \ {}^t \mathbf{1} T_s^k }{ \sigma^{2,s}_k}, 
\end{eqnarray*}
thus
\begin{eqnarray*}
 \mathbb{E}_{X | Y} [I^c(\delta,\mu)]\big|_{\theta=\widehat{\theta}}  &=&  \frac{\tau^s_k \ {}^t \mathbf{1} T_s^k }{ \widehat{\sigma}^{2,s}_k}
\end{eqnarray*}
\item for the $(\delta,\sigma^2)$ block, 
\begin{eqnarray*}
 \frac{\partial^2 }{\partial \delta \partial \sigma_k^{2,s}} \log{\mathbb{P}(Y,Z;\theta)}&=&  \frac{\partial }{\partial \delta}  Z^s_k  \left ( -\frac{n_s}{2  \sigma^{2,s}_k} +\frac{1}{2  \sigma^{4,s}_k} \| Y^s- \mu_{k}^{s} -T^s_k \delta\|_2^2 \right )  \\
 &=&  \frac{Z^s_k}{  \sigma^{4,s}_k} \left ( - \ {}^t (Y^s_t) T_s^k+ \mu_k^{s} \ {}^t \mathbf{1} T_s^k +\delta \ {}^t (T_s^k) T_s^k \right )  \\
 &=& - \frac{Z^s_k}{ \sigma^{4,s}_k} \left (  \ {}^t (Y^s_t) - \mu_k^{s}  -\delta \ {}^t (T_s^k) \right ) \ T_s^k
\end{eqnarray*}
thus
\begin{eqnarray*}
 \mathbb{E}_{X | Y} [I^c(\delta,\sigma^2)]\big|_{\theta=\widehat{\theta}}  &=& 
\frac{\tau^s_k}{ \widehat{\sigma}^{4,s}_k} \left (  {}^t(Y^s_t) - \widehat{\mu}_k^{s}  -\widehat{\delta} \ {}^t (T_s^k) \right )T_s^k
 \end{eqnarray*}
\end{itemize}
\end{appendices}

\section{Competing interests}
No competing interest is declared.

\section{Author contributions statement}

V. Brault, E. Lebarbier, and A. Rosier developed the statistical method and conducted the simulation study. All authors worked on the application on real data, the results of which were analyzed by V. Stoppin-Mellet.

\section{Acknowledgments}

This research has been conducted within the FP2M federation (CNRS FR 2036). All the simulations were carried out on the GRICAD infrastructure (\url{https://gricad.univ-grenoble-alpes.fr}), which is supported by the Grenoble scientific community. This work has been partially supported by MIAI@Grenoble Alpes (ANR-19-P3IA-0003).

\bibliographystyle{abbrvnat}
\bibliography{sample}

\end{document}